\documentclass[11pt]{article}

\usepackage{amsmath}
\usepackage{amssymb}
\usepackage[colorlinks=true,citecolor=blue]{hyperref}
\usepackage{amsthm}
\usepackage{algorithm}
\usepackage{algorithmic}
\usepackage{color}
\usepackage{soul}
\usepackage{graphicx}
\usepackage{paralist}
\usepackage{enumitem}
\usepackage{xspace}
\usepackage{wrapfig}

\def\denseformat{
\setlength{\textheight}{9in}
\setlength{\textwidth}{6.5in}
\setlength{\evensidemargin}{-0.2in}
\setlength{\oddsidemargin}{-0.2in}
\setlength{\headsep}{10pt}
\setlength{\topmargin}{-0.3in}
\setlength{\columnsep}{0.375in}
\setlength{\itemsep}{0pt}
}
\newtheorem{thm}{Theorem}[section]
\newtheorem{lem}[thm]{Lemma}

\newtheorem{prop}[thm]{Proposition}
\newtheorem{claim}{Claim}

\newtheorem{dfn}[thm]{Definition}

\newcommand{\dft}[1]{{\bfseries\boldmath\itshape{#1}}}

\newcommand{\ncdots}{\mkern-3mu\cdot\mkern-3mu\cdot\mkern-3mu\cdot\mkern-3mu}

\newcommand{\N}{\mathbf{N}}
\newcommand{\R}{\mathbf{R}}

\newcommand{\calA}{\mathcal{A}} % A family of activated intervals
\newcommand{\calB}{\mathcal{B}}
\newcommand{\calI}{\mathcal{I}}
\newcommand{\calP}{\mathcal{P}} % A family of packets

\renewcommand{\phi}{\varphi}

\newcommand{\level}{\mathrm{lv}}

\DeclareMathOperator{\Path}{Path}

\renewcommand{\th}{{\textrm{-th}}}

\newcommand{\PTS}{\textsf{PTS}\xspace}
\newcommand{\PPTS}{\textsf{PPTS}\xspace}
\newcommand{\HPTS}{\textsf{HPTS}\xspace}
\newcommand{\FormPaths}{\textsf{FormPaths}}
\newcommand{\ActivatePreBad}{\textsf{ActivatePreBad}}

\newcommand{\abs}[1]{\left|#1\right|}
\newcommand{\floor}[1]{\left\lfloor#1\right\rfloor}
\newcommand{\ceil}[1]{\left\lceil#1\right\rceil}
\newcommand{\paren}[1]{\left(#1\right)}
\newcommand{\set}[1]{\left\{#1\right\}}
\newcommand{\Set}[1]{\set{#1}}
\newcommand{\sucht}{\,:\,}
\newcommand{\Downto}{\textbf{downto}\xspace}

\newcommand{\Seq}[1]{\left\langle#1\right\rangle}

\denseformat\renewcommand{\paragraph}[1]{\par\noindent\textbf{#1}}
\title{\bf With Great Speed Come Small Buffers:\\
  Space-Bandwidth Tradeoffs for Routing %to Multiple Destinations
}
\author{Avery Miller\\ University of Manitoba%\\$-19^\circ$
\and
Boaz Patt-Shamir\\ Tel Aviv University%\\$16^\circ$
\and
Will Rosenbaum\\Max Planck Institute for Informatics%\\$2^\circ$
}
\begin{document}
%\begin{titlepage}
\maketitle
\begin{abstract}
  We consider the Adversarial Queuing Theory (AQT) model, where packet arrivals are subject to a maximum average rate $0\le\rho\le1$ and burstiness $\sigma\ge0$. In this model, we analyze the size of buffers required to avoid overflows in the basic case of a path.  Our main results characterize the space required by the average rate and the number of distinct destinations: we show that $O(k d^{1/k})$ space suffice, where $d$ is the number of distinct destinations and $k=\floor{1/\rho}$; and we show that $\Omega(\frac 1 k d^{1/k})$ space is necessary.  For directed trees, we describe an algorithm whose buffer space requirement is at most $1 + d' + \sigma$ where $d'$ is the maximum number of destinations on any root-leaf path.
\end{abstract}

\section{Introduction}

\paragraph{Background.}
Routing, be it physical (cars, trains, air traffic) or virtual (data
packets, voice call circuits, multicasts) has long been recognized as
a critical component of any system that allows for interaction and
communication. The introduction of packet
networks and store-and-forward routing as the basic technique of the
Internet (when it was still called
DARPANET~\cite{doi:10.1002/net.3230030202}) brought to the forefront
an  interesting parameter, namely the buffer space at nodes.

Buffers are, in some sense, the universal glue that connects different
components of a system smoothly. Intuitively, buffers allow each part
of a system to run at its own pace. This is necessary even in tightly
synchronized systems, because local processing speed depends on local
resources and local load, which are rarely uniform across the system.
Obviously, buffers cannot make up for prolonged gaps between the rates
of demand and capacity, but they help in overcoming short bursts by
smoothing the demand over time. This partial decoupling of input and
output makes buffer control and analysis a tricky issue.
Andrews et al.~\cite{DBLP:journals/siamcomp/AndrewsFHLZ00}
give strong results regarding the required buffer space and maximal
latency, but assuming that all packets are available at start, and 
ignoring initial storage. 

Typically, buffers are simple and easy to use so long as the load is
light. However, problems arise once loads increase and buffers are
filled: packet discards due to full buffers are usually viewed as a
failure.\footnote{TCP, the Internet's prevalent transport protocol, is
  an exception in this respect. TCP in fact forces periodic packet
  drops, as a part of its ongoing quest to grab free bandwidth. Since
  no explicit information about utilization is assumed, the protocol
  keeps increasing demand until an indication of a packet drop is
  received.} 
From the theoretical viewpoint 
(possibly because packet
drops greatly complicate stochastic analysis \`a la queuing-theory),
the issue of dropped packets was largely ignored by analytic studies
in the Theory of Networks. The introduction of competitive analysis
allowed for investigating the issue from the throughput
viewpoint~\cite{doi:10.1137/S0097539701399666}, gaining new insights
into the problem (see, e.g., the survey~\cite{Goldwasser:2010}).  But
by and large, the positive results concerned a single device (a switch
or a router), and results for even slightly more complex systems were
mostly negative~\cite{DBLP:conf/spaa/KesselmanMLP03}, in the sense
that in the examined cases, the competitive ratio of any algorithm
could be unacceptably large for practical applications.

A different approach was advocated by Adversarial Queuing Theory
(AQT), introduced in~\cite{Borodin2001}. AQT considers the scenario in
which each packet has a predetermined route, and forwarding is
\emph{greedy}. That is, if there is a packet buffered for a link $e$
at time $t$, then some packet will be forwarded over $e$ at time
$t$. A \emph{scheduling policy} is a rule which selects {which}
packet each buffer forwards each round. Classical AQT addresses the 
following
qualitative question: Which scheduling policies ensure that all buffer
occupancies are bounded irrespective of the duration of execution?

To have a meaningful answer for the above question,  one must
somehow restrict the demand presented by new packets. The
criterion used by AQT (adopting ideas of Cruz~\cite{Cruz1991})
is, roughly, that the set of routes associated with packets that
appear together (``injected by the adversary'') at any step is
edge-disjoint. This ensures that the raw bandwidth required by packets
does not exceed the available bandwidth. (In fact, overlap is 
allowed, but not more than some $\sigma\ge0$; a second parameter, 
$0\le\rho\le1$ bounds the maximum rate in which any 
edge is requested. See Def.~\ref{dfn:sigma-rho}.)

AQT was fruitful in identifying strengths and weaknesses of different
(greedy) scheduling policies (e.g.,~\cite{DBLP:journals/siamcomp/BhattacharjeeGL04}). However, the theory
was limited in two important respects. First,
restricting forwarding policies to be greedy turned out to be  a
serious handicap in terms of
efficiency~\cite{Aiello2003}. Secondly, the qualitative-only
view of the required buffer space (either bounded or infinite) is
quite coarse. For example, some protocols may be stable, but with space
requirement which is exponential in the number of nodes. Notable
exceptions to this state of affairs are the work of Adler and
Ros\'en on DAGs~\cite{Adler2002}
% where it was shown that on DAGs, the size of the
% buffers can be bounded by the length of the longest path to avoid
% overflows  
and the work by Ros\'en and
Scalosub on lines~\cite{DBLP:journals/talg/RosenS11}. 
% where it was shown that on DAGs, the size of the
% buffers can be bounded by the length of the longest path to avoid
% overflows
The relative
simplicity of the topologies in these studies is not coincidental. The
difficulty of understanding maximum buffer size \emph{quantitatively}
drove research to examine simple cases first: Even the case of a line 
%where all packets have the
%same destination 
turned out to be fairly non-trivial (see 
\cite{even_et_al:LIPIcs:2016:6352} for a recent result).

\paragraph{Recent progress: Single-destination trees.}
Recently, several works have addressed the problem of quantifying the
required buffer size under a $(\rho,\sigma)$-restricted packet
injections, but without the requirement that the forwarding protocol
must be greedy, and with an eye on \emph{locality}. In~\cite{Miller2016}, it was
shown that in the case of a single-destination tree, a centralized
algorithm can route any instance using $O(\rho+\sigma)$ buffer space
at every node. In~\cite{DBLP:conf/spaa/DobrevLNO17}
and~\cite{DBLP:conf/podc/Patt-ShamirR17} it is shown that
$\Theta(\rho\log n+\sigma)$ buffer space are necessary and sufficient for
protocols with constant locality. That result was later extended
in~\cite{PR-19} to show that $\Theta(\rho\ceil{\frac{\log n}r}+\sigma)$ is
necessary and sufficient for protocols with locality $r$. An interesting
consequence of this result is that $O(\log n)$ locality is sufficient
to perform as well as the best (offline)
strategy. % (for single destination trees).

\paragraph{Our results: Multiple destinations.}
%While the case of a single-destination tree is interesting (most 
%routing protocols produce such a tree for each destination), the 
%multi-destination case is obviously more useful.
In~\cite{DBLP:conf/podc/Patt-ShamirR17} it was shown that on a 
line, $\Omega(d)$ space is necessary to avoid overflows if there are 
$d$ distinct destinations, and if the average injection rate satisfies 
$\rho>1/2$. In this paper we show that the latter condition was no
accident: On one hand, we show that
if the average rate satisfies 
$\frac1{\ell+1}<\rho\le\frac{1}{\ell}$ for some integer $\ell$, then 
there is a centralized online algorithm which routes all packets using  
buffer space $O(\ell d^{1/\ell}+\sigma)$; and on the other hand, we 
prove a 
lower bound of  $\Omega(\frac 1 \ell d^{1/\ell}+\sigma)$ buffer space required
by any (offline) algorithm, for any integer $\ell>1/\rho$.
In particular, this result means that if $\rho\le\frac1{\log d}$, 
$O(\log d)$ buffer space suffices. 
%make another significant step toward the 
%general case of multiple destinations (and show that the requirement 
%$\rho>1/2$ was essential to the lower bound). While still assuming a 
In addition, for the directed tree topology, we prove that buffers of 
size $1 + d' + 
\sigma$ 
are sufficient, where $d'$ is an upper bound on the number of  
destinations along any leaf-root path. 
%For the case 
%of the path, we present a much tighter tradeoff between the injection 
%rate, number of destinations, and the buffer space requirement. 

\paragraph{Implications and open problems.}
One way to interpret our positive result is the following: Suppose
that in a given line system, some buffer space suffices to avoid packet
drops. If the system is modified so that the number of destinations is
increased by a factor $\alpha>1$ without changing the overall offered
load per link, overflows can be avoided by either increasing the
buffer sizes by factor $\alpha$, or by increasing both buffer space
and link bandwidths by an $O(\log\alpha)$ factor. This result corresponds
nicely to previous results such as online virtual circuit admission
control~\cite{DBLP:conf/focs/AwerbuchAP93,DBLP:journals/fttcs/BuchbinderN09}
exhibiting $O(\log n)$ competitiveness when demands are at most
$O(1/\log n)$ of the capacity. 
%Our results may further be used to
%complement R\"acke's result~\cite{DBLP:conf/focs/Racke02} of finding
%routes with polylogaritmic overlap.

Great challenges remain open on our way to fully understand the
general case. These include finding a decentralized (local) algorithm
for the multi-destination case, and extending the restricted topologies we consider
to general topologies. Regarding the former, we expect the
OED algorithm~\cite{DBLP:conf/spaa/DobrevLNO17,
  DBLP:conf/podc/Patt-ShamirR17} to be useful as it was for algorithms
of parameterized locality~\cite{PR-19}. Regarding the latter, we note that
the case of a union of trees is also important, due to the
fact that this topology is the output of many routing algorithms.

 \paragraph{Paper organization.} We present the algorithm gradually,
 from the simplest to the most general case: In Sec.~\ref{sec:basic-alg}
 we present algorithms for any utilization parameter $\rho$, and also 
 extend them to directed trees. The algorithm is then
 generalized in Sec.~\ref{sec:hierarchical-alg} to be hierarchical, 
 with $\floor{\frac1\rho}$ hierarchy levels. The lower bound is 
 presented in 
 Sec.~\ref{sec:lb}.  
 Sec.~\ref{sec:preliminaries} formalizes the model and defines some
 concepts and notation. Due to lack of space, many proofs appear 
 in an appendix.

\section{Preliminaries}
\label{sec:preliminaries}
Given a natural number $n$, we use the notation $\Seq 
n=\Set{0,\ldots,n-1}$.
We model a network as a directed graph $G = (V, E)$, where 
packets flow along the direction of edges.
For the most part, in this paper we restrict attention to the case 
where $G$ is a path, and take $V=\Seq{n}$, $E = \set{(i, i+1) \sucht 
0\le i <n-1}$. 
% determines the direction of pacekt movement. 

An execution of an algorithm proceeds in synchronous rounds. Each round 
consists of two steps: an \dft{injection step}, during which new 
packets arrive, and a \dft{forwarding step}, during which the algorithm 
performs computations and forwards packets. In this paper we assume 
that in each round, at most one packet can be forwarded over each link. 
When we refer to a state at 
round $t$, %e.g., $L^t(v)$, 
we compute the state \emph{after} the 
injection step, but \emph{before} the forwarding step. We denote the 
time in round $t$ immediately after the forwarding step by $t+$. 
E.g., 
$L^t(v)$ denotes the state of the buffer at $v$ 
after packets are injected and 
before forwarding; and $L^{t+}(v)$ denotes the state after forwarding, 
but 
before round $t+1$ packets arrive.

%\boaz{here?!}
%In order to determine which buffers forward packets in round $t$, a 
%buffer $L^t(v)$ can be \dft{active} or \dft{inactive}. In each round, 
%all buffers are initially inactive. During the forwarding step, our 
%algorithms select buffers to activate, then all active buffers 
%simultaneously forward. We denote the family of activated buffers by 
%$\calA$.

A node may partition its buffer into several parts, which we refer to 
as \dft{pseudo-buffers}. We use the notation $L_{k}^t(v)$ and $L_{j, 
k}^t(v)$ to denote pseudo-buffers, where $j$ and $k$ are parameters 
described below. 
%When a partitioned buffer becomes activated, the activation specifies 
%which pseudo-buffer is active, thereby determining which packet is 
%forwarded. 
We make no assumptions about the priority assigned to packets within a 
single (pseudo)-buffer, but for concreteness it is convenient to assume 
that all (pseudo)-buffers use Last-In, First-Out (LIFO) priority.

Throughout an execution, packets arrive in the network, controlled by 
an 
adversary. 
An \dft{adversary} or \dft{injection pattern} is a set $A$ of packets. 
A packet $P$ is represented by a triple $P = (t, i_P, w_P)$ where $t 
\in \N$ is injection round of $P$, $i_P\in V$ is $P$'s injection site,  
and $w_P\in V$ is 
$P$'s destination. We denote the unique path from $i_p$ to $w_P$
by $\Path(i_p, w_P)$ . 

We say that $P$'s path \dft{contains} a buffer $v$ if $v \in \Path(i_P, 
w_P)$. For an adversary $A$, a set of rounds $T$, and a buffer $v$, we 
denote the number of packets injected during $T$ whose trajectories 
contain $v$ by $N_T(v)$. That is,
\(
N_T(v) = \abs{\set{(t, i_P, w_P) \in A \sucht v \in \Path(i_P, w_P) \text{ and } t \in T}}.
\)

\begin{dfn}
	\label{dfn:sigma-rho}
  For any $\rho \in \R^+$ and $\sigma \in \N$ and adversary $A$, we say 
  that $A$ is \dft{$(\rho, \sigma)$-bounded} if for all buffers $v$ and 
  intervals of time $I$, we have
  \(
  N_T(v) \leq \rho \abs{I} + \sigma.
  \)
\end{dfn}

Intuitively, an adversary $A$ is $(\rho, \sigma)$-bounded if the 
average rate of packets needing to cross any particular 
edge $i$ is at most $\rho$. This average may not be exceeded by 
more than $\sigma$ packets when taken over any contiguous interval of 
time. The term 
$\sigma$ thus bounds the ``burstiness'' of $A$. 

%The following 
%definition %of ``excess'' of a buffer 
The term ``excess'' gives a measure of how much of the 
adversary's burst budget has been expended. % at any time $t$.

\begin{dfn}
  \label{dfn:excess}
  Fix an adversary $A$, rate $\rho$, round $t$, and buffer $v \in V$. The \dft{excess} of $v$ at time $t$, denoted $\xi^t(v)$, is 
  given by
  \begin{equation}
    \label{eqn:excess}
    \xi^t(v) = \max_{s \leq t} (\set{N_{[s, t]}(v) - \rho (t - s + 1)} \cup \set{0}).
  \end{equation}
\end{dfn}

The following lemma, which applies to graphs of any topology,
gives basic properties of excess.
% of each buffer is bounded by 
%$\sigma$ when $A$ is $(\rho, \sigma)$-bounded, and relates the change 
%in excess to the number of packets injected each round needing to 
%cross 
%any given edge. We note that Lemma~\ref{lem:excess} applies to 
%arbitrary graphs. 
The proof appears in Appendix \ref{app:preliminaries}.

\begin{lem}
  \label{lem:excess}
  Suppose $A$ is $(\rho, \sigma)$-bounded. Then for all rounds $t$ and buffers $v \in V$, we have
  \begin{inparaenum}[(1)]
  \item $\xi^t(v) \leq \sigma$, and
  \item $N_{\{t\}}(v) \leq \xi^t(v) - \xi^{t-1}(v) + \rho$.
  \end{inparaenum}
\end{lem}

%\begin{rem}
%  \label{rem:excess}
%  While we prove most of our results for the case where $G$ is the 
%directed path, the proof of Lemma~\ref{lem:excess} makes no reference 
%whatsoever to $G$'s topology. Thus, the conclusion of the lemma holds 
%for arbitrary graphs as well.
%\end{rem}

In what follows, we sometimes consider algorithms that only accept 
packets every few rounds. Specifically, for some fixed $\ell$ and 
every $k \in \N$, the algorithm treats all packets injected at times $t 
= (k - 1) \ell + 1, (k - 1) \ell + 2, \ldots, k \ell$ as being injected 
at time $k \ell$. The algorithm then performs a single forwarding step 
at time $k  \ell$,  waits until $(k + 1) \ell$ to perform its next 
forwarding step, etc. %We call the injection pattern formed by grouping 
%together packets injected in rounds $(k - 1) \ell + 1, \ldots, k \ell$ 
%the $\ell$-reduction of $A$. More formally:

\begin{dfn}
  \label{dfn:ell-reduction}
  Let $A$ be an adversary and $\ell$ a positive integer. The \dft{$\ell$-reduction} of $A$, denoted $A_\ell$ is the injection pattern defined by
  \(
  A_{\ell} = \set{\left(\floor{\frac{t - 1}\ell} + 1, i_P , w_P\right) 
  \sucht (t, i_P, 
  w_P) \in A}
  \).
\end{dfn}

The following lemma (proven in the appendix) characterizes $A_\ell$ using the parameters of $A$.

\begin{lem}
  \label{lem:ell-reduction}
  Suppose $A$ is $(\rho, \sigma)$ bounded and $\ell$ a positive integer. Then $A_\ell$ is $(\ell \cdot \rho, \sigma)$ bounded.
\end{lem}

\section{The Basic Algorithms}
\label{sec:basic-alg}
In this section, we describe a simple algorithm that requires buffers 
of size $O(d + \sigma)$ for all rates $\rho \leq 1$ and $d$ 
destinations. By Thm.~\ref{thmLB}, this space requirement is optimal 
for rates $\rho$ satisfying $1/2 < \rho \leq 1$. We break the 
description into two parts.  We first consider the case where all 
packets have a common destination: Algorithm ``Peak-to-Sink'' (\PTS) 
solves it using the best possible buffer space of $O(1 + \sigma)$ 
packets.  We then reduce the case of multiple destinations to single 
destinations in Algorithm ``Parallel Peak-to-Sink'' (\PPTS).
Finally we extend the result to directed trees.

%\boaz{here?!}
\emph{Implementation convention.}
In order to determine which buffers forward packets in round $t$, a 
buffer $L^t(v)$ can be \dft{active} or \dft{inactive}. In each round, 
all buffers are initially inactive. During the forwarding step, our 
algorithms select buffers to activate, then all active buffers 
simultaneously forward. We denote the family of activated buffers by 
$\calA$.

\subsection{Peak to Sink (PTS) forwarding}
\label{sec:p2s}

\begin{wrapfigure}[9]{r}{.45\textwidth}
  \begin{minipage}{.4\textwidth}\vspace*{-8mm}
    \begin{algorithm}[H]\small
      \caption{$\PTS(w)$}
      \label{alg:peak-to-sink}				
      \begin{algorithmic}[1]
        \FORALL{rounds $t$}
        \STATE $i \leftarrow $ left-most buffer with 
        $\abs{L^t(i)} \geq 2$
        \STATE $\calA \leftarrow [i, w-1]$
        \STATE forward all buffers in $\calA$
        \ENDFOR
      \end{algorithmic}
    \end{algorithm}
  \end{minipage}
\end{wrapfigure}

Assume first that all packets $P \in A$ have a common destination $w$, 
which we take to be $w = n - 1$ without loss of generality. Consider a 
configuration $L^t$. We say that a buffer $i \in \Seq{n}$ is \dft{bad} 
if $\abs{L^t(i)} \geq 2$. $\PTS$ (Alg.~\ref{alg:peak-to-sink}) is 
simple: if $i$ is the left-most bad buffer, then all non-empty buffers 
$i' \geq i$ forward.% in the $t\th$ forwarding step. 
$\PTS$ requires the minimal possible buffer space, as we claim next 
(see Appendix \ref{app:basic-alg} for a proof).

\begin{prop}
  \label{prop:pts}
  Suppose $A$ is a $(\rho, \sigma)$-bounded adversary with $\rho \leq 1$ and all packets $P \in A$ have destination $w$. Then the maximum buffer size is at most $2 + \sigma$.
\end{prop}

\subsection{Parallel Peak to Sink (PPTS) forwarding}
\label{sec:pp2s}

Assume now that there are $d$ destinations. Let $W = \set{w_0, w_1, 
\ldots, w_{d-1}}$ denote the set of destinations, where $w_0 < w_1 < 
\cdots < w_{d-1}$. The idea of the algorithm is for each $w \in W$ to 
treat all packets with destination $w$ as an instance of the single 
destination algorithm. Capacity constraints preclude us from performing 
all $d$ instances in parallel every round, so $\PPTS$ chooses an 
appropriate maximal set of (pseudo)-buffers that can simultaneously 
forward.

To describe $\PPTS$ formally, we view each buffer $i$ as a collection 
of $d$ pseudo-buffers, where the $k\th$ pseudo-buffer in $i$ stores 
packets destined for $w_k$.\footnote{This is known in practice as 
``virtual output queuing'' (e.g., 
\cite{DBLP:journals/micro/McKeownIMEH97}).  } We denote the contents of 
$i$'s $k\th$ pseudo-buffer by
\[
L_k^t(i) = \set{P \in L^t(i) \text{ and } P \text{ has destination } 
w_k}.
\]
For any destination $w_k$, we refer to pseudo-buffers $L_k(i)$ as \dft{$k$-pseudo-buffers}. We say that a pseudo-buffer $L_k^t(i)$ is \dft{bad} at time $t$ if $\abs{L_k^t(i)} \geq 2$. In this case we call $L^t(i)$ \dft{bad for $k$}. We call a packet $P$ in a $k$-pseudo-buffer \dft{$k$-bad} if it is stored in a position $p \geq 2$ in $L_k^t(i)$.

$\PPTS$ (Alg.~\ref{alg:pp2s}) works as follows. Starting from the 
right-most destination, $w_{d-1}$, the algorithm searches for the 
left-most 
bad $(d-1)$-pseudo-buffer, if any. If a bad $(d-1)$-pseudo-buffer 
$L_{d-1}^t(i_{d-1})$ 
is found, all $(d-1)$-pseudo-buffers in the range $[i_{d-1}, w_{d-1} - 
1]$ are 
activated. The algorithm then searches for bad $d-2$ pseudo-buffers. 
Suppose $L_{d-2}^t (i_{d-2})$ is the left-most $(d-2)$-pseudo-buffer. 
If $i_{d-2} < i_{d-1}$, then all $(d-2)$-pseudo-buffers in the interval 
$[i_{d-2}, i_{d-1} - 1]$ are activated; otherwise, no $(d-2)$-pseudo-buffers 
are 
activated. This process is continued for destinations $w_{d-3}, 
w_{d-4}, \ldots, w_0$, where an interval of $k$-pseudo-buffers is 
activated only if a $k$-bad pseudo-buffer is encountered to the left of 
all pseudo-buffers activated thus far. By construction, the intervals 
activated for each destination are disjoint, so forwarding can be 
implemented without violating capacity constraints.

\begin{wrapfigure}{r}{.47\textwidth}
	\begin{minipage}{.47\textwidth}\vspace{-6mm}
\begin{algorithm}[H]\small
  \caption{$\PPTS(I, W)$  (do each round)}
  \label{alg:pp2s}
  \begin{algorithmic}[1]
    \STATE $W = \set{w_0 < w_1 < \cdots < w_{d-1}}$
    \hfill\COMMENT{destinations}
    \STATE $i \leftarrow w_d$, $\calA \leftarrow \varnothing$
    \FOR{$k \gets d-1$ \Downto $0$}
    \IF{there exists $i' < i$ such that $i'$ is bad for $w_k$}
    \STATE $i_k \leftarrow \min\set{i' \sucht L_k(i') \text{ is bad}}$ 
    \STATE $\calA \leftarrow \calA \cup \set{L_k(i) \sucht i \in [i_k, i-1]}$
    \STATE $i \leftarrow i_k$
    \ENDIF
    \ENDFOR
    \STATE each nonempty $L_k(i) \in \calA$ forwards
  \end{algorithmic}
\end{algorithm}
\end{minipage}\end{wrapfigure}

%\begin{rem}
%  \label{rem:ppts}
%  The 
We note that algorithm $\PPTS$ need not be told the set of destinations 
$W$ to which an adversary $A$ injects in advance. Instead, we can 
assume that all nodes $i \in [1, n-1]$ are potential destinations, 
giving rise to $n$ pseudo-buffers per node. However, if an adversary 
only injects packets with $d$ distinct destinations, then at most $d$ 
pseudo-buffers in any given buffer will be non-empty. 
%\end{rem}

Our analysis shows that for any $(\rho, \sigma)$-bounded adversary $A$, $\PPTS$ maintains the following invariant: For each buffer $i$, the number of bad packets that must cross $i$---i.e., the number of bad packets in buffers $i' \leq i$ with destinations $w_k > i$---is at most one plus the excess $\xi^t(i) + 1 \leq \sigma + 1$. In particular, every buffer $L(i)$ contains at most $1 + \sigma + d$ packets, as it can contain at most $d$ packets that are not bad. The following proposition states the key property of \PPTS.  

%% First, we show that the forwarding pattern prescribed by $\PPTS$ is feasible, in the sense that at most one pseudo-buffer in a buffer $i$ forwards in each round.
%% The proof is deferred to the appendix.
%% \begin{lem}
%%   \label{lem:ppts-feasible}
%% %  Fix a configuration $L$ and consider a single step of $\PPTS$. For 
%% %each $k \in [d]$ let $I_k = [a_k, b_k]$ denote the (possibly empty) 
%% %interval of $k$-pseudo-buffers activated during the execution. Then 
%% %for 
%% %all non-empty $I_k, I_{k'}$ with $k < k'$, we have $b_{k} < a_{k'}$. 
%% %In 
%% %particular, the intervals created in a step of $\PPTS$ are pair-wise 
%% %disjoint, hence at most one pseudo-buffer is activated in each buffer.
%%   Consider a single step of $\PPTS$, and $I_k$ be a the (possibly 
%%   empty) interval 
%%   of 
%%   $k$-pseudo-buffers activated during that step. Then for any distinct
%%   $k,k'\in\Seq{d}$ we have $I_k\cap I_{k'}=\varnothing$.
%% \end{lem}

\begin{prop}
  \label{prop:ppts}
  Let $A$ be any $(\rho, \sigma)$-bounded adversary such that all packets have destinations in $W$ with $d = \abs{W}$. Then the maximum buffer usage of $\PPTS$ is at most $1 + d + \sigma$.
\end{prop}

To facilitate the proof, %(which appears in the appendix), 
we introduce the concept of ``badness,'' along with some new notation.

\begin{dfn}
  \label{dfn:parallel-badness}
  Let $i \in \Seq{n}$ be a buffer, and $w_k$ a destination such that $w_k > i$. Given $t$, we denote the number of $k$-bad packets in $i$ at time $t$ by $\beta_k^t(i)$. Formally, \( \beta_k^t(i) = \max\set{|L_k^t(i)| - 1, 0}.  \) The \dft{$k$-badness} of $i$, denoted $B_k^t(i)$, is the total number of bad packets in buffers $i' \leq i$ with destination $w_k$. That is, \( B_k^t(i) = \sum_{i' \leq i} \beta_k^t(i').  \) The \dft{badness} of $i$, denoted $B^t(i)$, is the total number of bad packets in buffers $i' \leq i$ with destinations $w > i$. Thus, \( B^t(i) = \sum_{k : w_k > i}B_k^i(t).  \)
\end{dfn}

(We stress that $B_k^t(i)$ and $B^t(i)$ count also badness of packets 
upstream from $i$.) In the proof of Prop.~\ref{prop:ppts}, we use the 
following key lemma (proof in Appendix \ref{app:basic-alg}).
\begin{lem}
  \label{lem:bad-decrease}
  Let $L$ be a configuration of $\Seq{n}$ with destinations $W = 
  \set{w_0 < w_1 < \cdots < w_{d-1}}$. Suppose that for some $k \in 
  \Seq d$, $I_k = [a_k, b_k] \subseteq \Seq{n}$ is an interval such 
  that $\abs{L_k(a_k)} \geq 2$ and $b_k < w_k$. Let $L^+$ be the 
  configuration obtained by forwarding only from nonempty 
  $k$-pseudo-buffers in $I_k$. For $i \in \Seq{n}$ let $B_k(i)$ and 
  $B_k^+(i)$ be the $k$-badness of $i$ before and after forwarding, 
  respectively. Similarly, $B(i)$ and $B^+(i)$ denote the badness of 
  $i$ before and after forwarding. Then
  \[
  \begin{array}{lcl}
    B_k^+(i) \leq
    \begin{cases}
      B_k(i) - 1 &\text{if } i \in I_k\\
      B_k(i) &\text{if } i \notin I_k~,
    \end{cases}
    &
    \qquad\text{and similarly,}\qquad&
    B^+(i) \leq
    \begin{cases}
      B(i) - 1 &\text{if } i \in I_k\\
      B(i) &\text{if } i \notin I_k~.
    \end{cases}
  \end{array}
  \]
\end{lem}

\subsection{Extension to Trees}
In Appendix~\ref{app:trees}, we explain how to extend both \PTS and 
\PPTS to the case of directed trees. We assume that all edges are
directed toward the root, and that packets follow directed paths.
In this setting we prove the following.
\begin{prop}
	\label{prop:tree-ppts}
	Let $A$ be any $(\rho, \sigma)$-bounded adversary such that all 
	packets have destinations in $W$. Then the maximum buffer usage of 
	$\PPTS$ is at most $1 + d' + \sigma$, where $d'$ is an upper bound 
	on the number of destinations in any leaf-root path.
\end{prop}

\section{Hierarchical Algorithm}
\label{sec:hierarchical-alg}
In this section, we describe Algorithm
$\HPTS$ (Hierarchical Peak-to-Sink), that achieves significantly better 
performance than $\PPTS$ in the case where the rate of the adversary 
satisfies $\rho \leq 1/2$. The algorithm description is parametrized by 
the number of \dft{levels} in the construction, denoted $\ell \in \N$. 
In the 
case $\ell = 1$, $\HPTS$ reduces to $\PPTS$. In general case, we obtain 
the following result. For the remainder of the section, we assume that 
$\rho$ and $\ell$ satisfy $\rho \cdot \ell \leq 1$.

\begin{thm}
  \label{thm:hpts}
  There exists an algorithm, $\HPTS$, such that for every positive integer $\ell$ and $(\rho, \sigma)$-bounded adversary $A$ such that $\rho \cdot \ell \leq 1$, the space requirement of $\HPTS$ is at most $\ell n^{1/\ell} + \sigma + 1$.
\end{thm}

\subsection{Partitioning the network}

Our algorithm is based on a hierarchical partition of the line.  For simplicity, we assume that $n=m^\ell$ for some integers $m$ and $\ell$. We consider the base-$m$ representation of buffer indices.  That is, given $i \in \Seq n$ we write the string $i=i_{\ell-1}i_{\ell-2}\cdots i_0$, where $i_j\in\Seq m$ for $j=0,\ldots \ell-1$ and $i=\sum_{j=0}^{\ell-1}i_jm^j$. We refer to $i_j$ as the $j\th$ \emph{digit}, or the digit in the $j\th$ \emph{position}, of $i$.%\will{We already have \emph{positions} in buffers---should we use another word here? Or we could use\emph{height} instead of position for locations within a buffer.}

Using this convention, for $j\in\Seq{\ell}$ and $r\in\Seq{m^{\ell-j-1}}$, define the intervals $I_{j, r}$ by
\[
I_{j, r} = \Set{rm^{j+1} + k \sucht 0\le k<m^{j + 1}}~.
\]
Intuitively, the interval $I_{j,r}$ consists of all $m^j$ nodes whose most significant $\ell-j-1$ digits have value $r$. For each fixed $j$, the set of intervals
\(
\calI_j = \set{I_{j, r} \sucht r \in \Seq{m^{j}}}
\)
forms a partition of $\Seq{n}$. We call $\calI_j$ the \dft{level-$j$} partition, and refer to each $I_{j,r} \in \calI_j$ as a \dft{level-$j$} interval.  For each $j > 0$, each level-$j$ interval contains exactly $m$ level-$(j-1)$ subintervals. We denote the union of the partitions by $\calI = \calI_0 \cup \calI_1 \cup \cdots \cup \calI_{\ell-1}$.

The idea of $\HPTS$ is to run an independent instance of $\PPTS$ in each interval of $\calI$, allowing each interval of level $j$ to have exactly $m$ (intermediate) destinations: one for each left endpoint of a subinterval of level $j-1$.  Intermediate destinations are chosen so as to ``correct'' the index of a packet's location to the index of its final destination, digit by digit: position $j$ of the index is corrected at level $j$. Formally, we have the following. 

\begin{dfn}
  \label{dfn:intermediate-dest}
  Let $i,w\in\Seq{n}$. Assume that $i<w$ and that the largest position 
  (in base-$m$ notation) in which $i$ and $w$ differ is $j$.  Then the 
  \dft{intermediate destination of $i$ to $w$} is 
  $x(i,w)=\floor{w/m^j}m^j$. The route $[i,x(i,w)]$ is called a 
  \dft{segment}, and its \dft{level}, denoted $\level(i, w)$, is $j$.

  A packet $P$ with destination $w$ residing in a buffer $L(i)$, \dft{crosses $i$ at level} $\level(i, P) = \level(i, w)$. If $j = \level(i, P)$, we define $P$'s \dft{level-$j$ intermediate destination} to be $w_{P, j} = x(i, w)$.
\end{dfn}

If the largest position in which $i$ and $w$ differ is $j$, then the 
segment from $i$ to $x(i,w)$ is contained in a level-$j$ interval 
$I_{j}$.
%\boaz{corrected from level $j$} \will{I think $j$ is correct with the 
%level numbering now. Level 0 intervals should have length $m$, not $1$}
If $j \geq 1$, then for all $j' \le j$, $x(i,w)$ is the left endpoint 
of some level-$j'$ interval. %\boaz{corrected from $<$}
If a packet is injected into buffer $i$ with destination $w$, we can think of its trajectory $[i, w]$ as being partitioned into segments $[i, x(i, w) - 1], [x(i, w), x(x(i,w), w)], \ldots$, where the level of these segments is strictly decreasing. Further, with the possible exception of the initial segment, the left endpoints of all segments are left endpoints of intervals $I \in \calI$.

\begin{figure}[t]
  \begin{center}
    \includegraphics[scale=0.7]{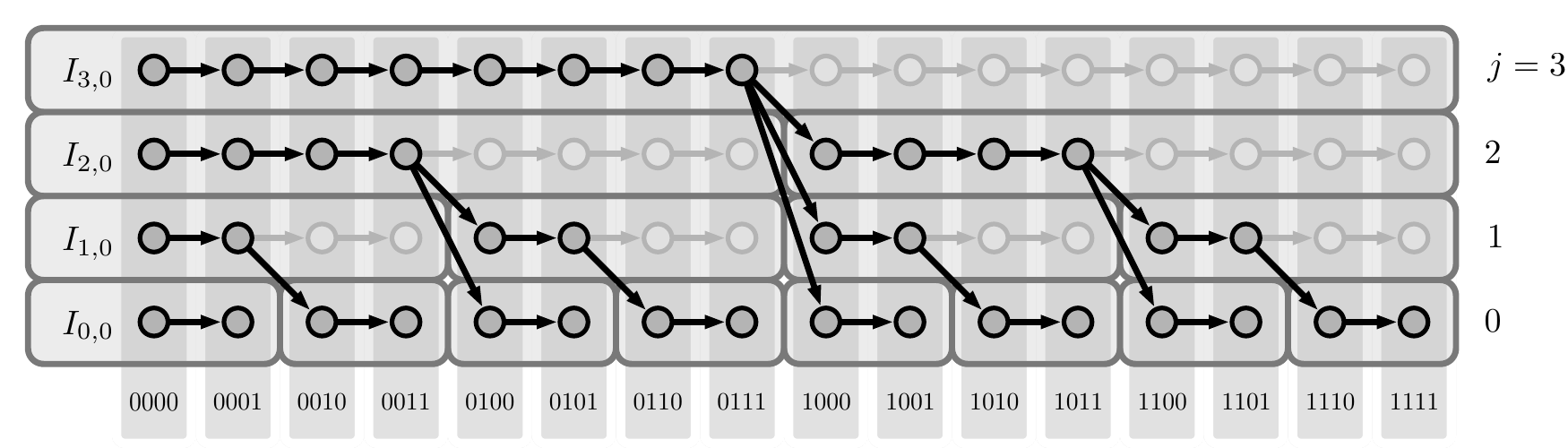}    
  \end{center}
  \caption{\it\small The network with $n = 16$, $m = 2$, and $\ell = 4$. Each column represents a single buffer, and each row represents a level. The horizontal boxes indicate divisions between intervals in $\calI$. Given a packet $P$ injected into $i$ with destination $w$, the virtual trajectory of $P$ is given by the unique path from a pseudo-buffer in $i$ to a level-$0$ pseudo-buffer in $w - 1$.}
  \label{fig:interval-tree}
\end{figure}

When $P$ is stored in a buffer $i$, we think of $P$ as virtually residing in the subinterval $I_{j}$ at level $j = \level(i, P)$ containing $i$. When $P$ is forwarded to its next intermediate destination, $w_{P, j} = x(i, w)$, its virtual location moves the corresponding interval at level $j' = \level(w_{P, j}, P)$. We emphasize that the \emph{virtual} location of each packet $P$ in a buffer $i$ (i.e., level and intermediate destination) is a function of $i$ and $P$'s destination. Thus, $i$ can compute the virtual location of each packet it stores locally. Figure~\ref{fig:interval-tree} illustrates the virtual motion packets through the network. Formally, the $\HPTS$ algorithm implements the virtual packet movement by partitioning each buffer into pseudo-buffers, as defined below.

%% \begin{algorithm}
%%   \caption{$\Route(a=a_{\ell-1}\cdots a_0,b=b_{\ell-1}\cdots b_0, j)$}
%%   \label{alg:route}
%%   \begin{algorithmic}[1]
%%     \IF{$a=b$}
%%     \RETURN empty string 
%%     \ENDIF
%%     \IF{$a_j=b_j$}
%%     \RETURN $\Route(a,b,j-1)$
%%     \ENDIF
%%     \STATE $x\gets \floor{b/m^j}\cdot m^j$
%%     \RETURN $(x,j)\circ\Route(x,b,j-1)$
%%     \hfill\COMMENT{concatenate segment $(a,x)$ of level $j$ to route's 
%%       remainder}
%%   \end{algorithmic}
%% \end{algorithm}

%% The way the route is computed is formalized in Algorithm~\ref{alg:route}, which, when called with $\Route(a,b,\ell-1)$, returns a sequence of intermediate destinations with their levels.  Packets progress from one intermediate destination (initially the source) to the next one in the appropriate level. Clearly, there are at most $\ell$ segments a complete source-destination route (equivalently, at most $\ell-1$ intermediate destinations), because segments levels are strictly decreasing.
%% \will{Do we need this algorithm? The level and intermediate destination of a packet $P$ stored in $L^t(i)$ are uniquely determined by $i$ and $P$'s destination. Since this is a purely local computation for $i$ (independent of $P$'s history, etc), I don't see any need to compute the route in advance.}

%% To facilitate forwarding, we partition each buffer $L(i)$ into $\ell \cdot m$ pseudo-buffers as follows.

\begin{dfn}
  \label{dfn:jk-pseduo-buffer}
  Let $t \in \N$, $i \in \Seq{n}$, $j \in \Seq{\ell}$ and $k \in\Seq m$. Let $I_j(i)$ be the level-$j$ interval containing $i$, and let $W_{j}(i) = \set{w_0 < w_1 < \cdots < w_{m-1}}$ denote the left endpoints of the level-$(j-1)$ intervals contained in $I_j(i)$.  We define the \dft{$(j,k)$-pseudo-buffer} of $i$ at time $t$, denoted $L_{j, k}^t(i)$, to contain all packets in $L^t(i)$ whose current segment level is $j$ and whose next intermediate destination is $w_k$.
\end{dfn}

Since $j \in\Seq{\ell}$ and $k\in\Seq{m}$, each buffer $L^t(i)$ is partitioned into $\ell m = \ell n^{1/\ell}$ pseudo-buffers. As in our description of the $\PPTS$ algorithm, we define a notion of badness for packets, pseudo-buffers, and buffers.

\begin{dfn}
  Let $t \in \N$, $i \in \Seq{n}$, $j \in \Seq{\ell}$, and $k \in 
  \Seq{m}$. We say that the $(j, k)$-pseudo-buffer $L_{j,k}^t(i)$ is 
  \dft{bad} at time $t$ if $|{L_{j, k}^t(i)}| \geq 2$. A packet in 
  $L_{j, k}^t(i)$ is called \dft{bad} if its position in $L_{j, 
  k}^t(i)$ is at least $2$.  The number of bad packets in 
  $L_{j,k}^t(i)$ is $\beta_{j, k}^t(i) = \max \set{|L_{j, k}^t(i)| - 1, 
  0}$.
  
   \end{dfn} 
   In the definition below, note that the badness of a node accounts 
   also for bad packets in other (upstream) nodes.
   \begin{dfn}
  
  Let $i\in\Seq n, j\in\Seq\ell, k\in\Seq m$ and let $I = [a, b]$ be a level-$j$ interval containing $i$. The \dft{$(j, k)$-badness} of a node $i$ at time $t$, denoted $B_{j,k}^t(i)$, is defined by
  \(
  B_{j, k}^t(i) = \sum_{i' = a}^i \beta_{j, k}^t(i').
  \)
  That is, $B_{j, k}^t(i)$ is the number of bad packets in buffers $i' \in I$ with $i' \leq i$ whose segment level at $i$ is $j$ and whose level $j$ intermediate destination is $w_k$. The \dft{level-$j$ badness} of $i$ is $B_j^t(i)= \sum_{k \in \Seq{m}} B_{j, k}^t(i)$, and the \dft{badness} of $i$ is $B^t(i)=\sum_{j \in \Seq{\ell}} B_j^t(i)$.
\end{dfn}

\subsection{Algorithm description}

%\hspace{-5pt}
\begin{minipage}[t]{.43\textwidth}
  \begin{algorithm}[H]\small
    \caption{$\HPTS(n, \ell, t)$ step}  \label{alg:hpts}
    \begin{algorithmic}[1]
      \STATE $\calA \leftarrow \varnothing$
      \hfill\COMMENT{active pseudo-buffers}
      \STATE $\lambda \leftarrow t \bmod \ell$
      \IF{$\lambda = 0$}\label{l:lump1}
      \STATE accept round $(t-\ell)$-,$\ldots,(t - 1)$-injections  
      %\hfill\COMMENT{compute routes by Alg.~\ref{alg:route}}
      \ENDIF\label{l:lump2}
      \FORALL{$r \in\Seq{m^{\ell-\lambda}}$}\label{l:ppts1}
      \STATE $\calA\gets\FormPaths( I_{\lambda, r}, \lambda)$
      \ENDFOR    \label{l:ppts2}
      \FOR{$j = \lambda - 1$ \Downto $0$}\label{prebad1}
      \STATE $\calA\gets\ActivatePreBad(\calA, j)$
      \ENDFOR\label{prebad2}
      \STATE each nonempty $L_{j, k}(i) \in \calA$ forwards a packet \label{l:fwd}
    \end{algorithmic}
  \end{algorithm}
\end{minipage}
%\vrule
\hspace*{.04\textwidth}
%\vrule
\begin{minipage}[t]{.53\textwidth}
  \begin{algorithm}[H]\small
    \caption{$\FormPaths(I,j)$}
    \label{alg:form-paths}
    \begin{algorithmic}[1]
      \STATE Let $W = \set{w_0 < w_1 < \cdots < w_{m-1}}$ be the 
      intermediate
      destinations in $I$
      %    \COMMENT{meaning: let $W$ be all left endpoints of 
      %level $j-1$ in 
      %$I$?}
      \STATE $\calA\gets\varnothing$; $i' \leftarrow w_{m-1}$
      \FOR{$k\gets m-1$ \Downto $0$}
      \IF{there exists $i < i'$ such that $i$ is \emph{bad} for $w_k$}
      \STATE $i_k \leftarrow \min\set{i'' \sucht |L_{j,k}^t(i'')| \geq 2}$
      \STATE $\calA \leftarrow \calA \cup \set{L^t_{j, k}(i) | i \in [i_k, \min\set{i' - 1, w_k - 1}]}$
      \STATE $i' \leftarrow i_k$
      \ENDIF
      \ENDFOR
      \RETURN $\calA$
    \end{algorithmic}
  \end{algorithm}	
\end{minipage}
\bigskip

The idea of Algorithm $\HPTS$ (Alg.~\ref{alg:hpts}) is as follows.  Let 
us refer to a sequence of rounds $t, t + 1, \ldots, t + \ell - 1$ with 
$t \equiv 0 \pmod \ell$ as a \dft{phase}. That is, the $\phi\th$ phase 
consists of rounds $(\phi - 1) \ell, \ldots, \phi \ell-1$. We assume 
that all packets injected during a phase $\phi$ are only accepted 
during the first round of phase $\phi + 1$ 
(Lines~\ref{l:lump1}--\ref{l:lump2}). 
%Lumping packets together in this  way gives rise to 
This results in $A_\ell$, i.e., the $\ell$-reduction of $A$ 
(cf.~Def.~\ref{dfn:ell-reduction}).

%% Whenever a packet is injected, its route is broken into segments by Algorithm~\ref{alg:route}. Each segment is handled by an independent instance of $\PPTS$ (with a small caveat, see below).

Due to capacity constraints, not all pseudo-buffers can simultaneously 
forward packets, so $\HPTS$ employs \emph{time-division multiplexing}: 
At any time $t$, only the level-$\lambda$ intervals are activated, 
where 
$\lambda \equiv t\bmod\ell$. Since same-level intervals are 
edge-disjoint, all intervals in $\calI_\lambda$ can be active in 
parallel (Lines~\ref{l:ppts1}--\ref{l:ppts2}).

We note that the family of buffers activated by 
Algorithm~\ref{alg:form-paths} is identical to the activation pattern 
in an iteration of $\PPTS$ in which~(1) only level-$j$ packets are 
considered, and~(2) the destination of each packet is taken to be its 
level-$j$ intermediate destination. The difference between $\HPTS$ and 
$\PPTS$ is in Lines~\ref{prebad1}--\ref{prebad2} of 
Alg.~\ref{alg:hpts}: Activating buffers in one level may cause the 
activation of buffers in lower levels. This is done to anticipate when 
a packet reaches an intermediate destination and switches level---and 
interval instance---where it may become bad in the new instance. 
Activating the new instance just before a new packet arrives to a bad 
position is sufficient to avoid an increase of badness for any buffer 
in a given round. Formally, we define the notion of pre-badness as 
follows.

\begin{wrapfigure}[12]{r}{.53\textwidth}
	\begin{minipage}{.53\textwidth}\vspace{-9mm}
		\begin{algorithm}[H]\small
			\label{alg:activate-pre-bad}
			\caption{$\ActivatePreBad(\calA, j)$}
			\begin{algorithmic}[1]
				\FORALL{$r \in \Seq{m^{\ell-j-1}}$}
				\STATE $I = [a, b] \leftarrow I_{j, r}$
				\IF{$\exists P$ pre-bad for $a$ at level $j$ and $a$ is 
				inactive}
				\STATE $k \leftarrow$ index s.t. $w_k = w_{P, j}$
				\STATE $w \leftarrow \max\set{i \in I \sucht i \leq w_k 
				\text{ and } [a, i] \text{ is 
				inactive}}$\label{l:inactive}
				\STATE $\calA \leftarrow \calA \cup \set{L_{j, k}^t(i) 
				\sucht i \in [a, w]}$
				\ENDIF
				\ENDFOR
				\RETURN $\calA$
			\end{algorithmic}
		\end{algorithm}		
	\end{minipage}
\end{wrapfigure}

\begin{dfn}
  \label{dfn:pre-bad}
  Suppose $L^t_{j,k}(i)$ is activated and let $P \in L^t_{j,k}(i)$ be the packet sent by $L^t_{j,k}(i)$. Suppose that $P$'s level-$j$ intermediate destination is $i + 1$ (i.e., $P$ is about to make the last hop in its level-$j$ segment), where it will be associated with pseudo-buffer $L_{j', k}(i')$ with $j'<j$. We say that the packet $P$ is \dft{pre-bad} if $|L_{j', k}(i+1)|\geq 1$. In this case, we also call the buffer $i+1$ \dft{pre-bad for $P$ at level $j'$}.
\end{dfn}

%\will{By the result of Adler and Ros\'{e}n, we can actually use greedy forwarding at level $\ell$ without affecting the asymptotic performance.}

\subsection{Algorithm analysis}

Our strategy in the analysis of $\HPTS$ is to bound the badness of each 
buffer $i$ by $\sigma$ at the end of each phase (just as our analysis 
of $\PPTS$ bounds the badness of each buffer at the end of each 
round).  Since each buffer contains $\ell \cdot m = \ell \cdot 
n^{1/\ell}$ pseudo-buffers, each buffer can contain at most $\ell \cdot 
n^{1/\ell}$ non-bad packets. Thus Theorem~\ref{thm:hpts} follows from 
giving an appropriate bound on the badness of each buffer. The basic 
argument is the same as our analysis of $\PPTS$: apply 
Lemma~\ref{lem:excess} to show that the increase in badness of a buffer 
is at most one more than the increase in excess of the buffer, then 
show that forwarding reduces the badness of the buffer by at least one 
(if it is positive). First, we show that the forwarding in 
Line~\ref{l:fwd} of $\HPTS$ is feasible (proof in Appendix 
\ref{app:hierarchical-alg}).

\begin{lem}
  \label{lem:hpts-feasible}
  Consider the set $\calA$ of active pseudo-buffers immediately before 
  Line~\ref{l:fwd} of $\HPTS$. Then $\calA$ is \emph{feasible},
  i.e., % in the sense that 
  for each $i \in \Seq{n}$, there is at most one active pseudo-buffer 
  $L_{j, k}^t(i) \in \calA$.
\end{lem}

Lemma~\ref{lem:hpts-bad-decrease} asserts that if a node has positive 
badness at the beginning of a phase, its badness strictly decreases 
during the phase. 
Thus, it is analogous to 
%Lemmas~\ref{lem:bad-decrease} and~\ref{lem:ppts-bad-decrease} used in 
the analysis of $\PTS$ and $\PPTS$ (see App.~\ref{app:basic-alg}). 
Thm.~\ref{thm:hpts} follows from an argument analogous to the proofs of 
Prop.~\ref{prop:pts} and Prop.~\ref{prop:ppts}.

\begin{lem}
  \label{lem:hpts-bad-decrease}
  Let $t = (\varphi - 1) \ell$ so that the $\phi\th$ phase consists of rounds $t + 1, t+2,\ldots, t + \ell$. Assume that no new packets arrive during rounds $t+2, t+3, \ldots, t + \ell$. Then the $\HPTS$ algorithm guarantees that for all $i \in \Seq{n}$, $B^{(t + \ell)+}(i) \leq \max\set{B^{t+1}(i) - 1, 0}$.
\end{lem}

The proof of Lemma~\ref{lem:hpts-bad-decrease} (in 
appendix~\ref{app:hierarchical-alg}) follows from two claims. First, in 
each round of a phase, the badness of each node is non-increasing. 
Second, if $B^{t+1}(i) > 0$, then $B(i)$ strictly decreases during some 
round of the phase. Essentially we show that if $B_{\lambda}(i) > 0$ 
during the $\lambda\th$ round of a phase, then either $B(i)$ decreases, 
or a bad packet (for $i$) at level $\lambda$ is replaced with a bad 
packet at some level $j < \lambda$ (that was pre-bad before 
forwarding). Since levels are activated in decreasing order over the 
course of a phase, $B(i)$ strictly decreases in some phase.

\begin{proof}[Proof of Theorem~\ref{thm:hpts}]
  Let $A$ be a $(\rho, \sigma)$-bounded adversary such that $\rho \cdot \ell \leq 1$, and let $A_\ell$ be the $\ell$-reduction of $A$ (Definition~\ref{dfn:ell-reduction}). For each $\phi$, let $t(\phi) = (\varphi - 1) \cdot \ell + 1$ denote the beginning of the $\phi\th$ phase, and let $t(\phi+)$ denote the time at the end of phase $\phi$---that is, after the forwarding step of round $t(\phi) + \ell - 1$.
  
  We claim that for each $\phi$ and buffer $i \in \Seq{n}$, we have $B^{t(\phi)}(i) \leq \xi^{t(\phi)}(i) + 1$ and $B^{t(\phi)+}(i) \leq \max\set{B^{t(\phi)}(i) - 1, 0}$. The argument follows by induction. For the base case, both inequalities hold trivially, as no packets have been accepted. For the inductive step, suppose the equations above hold for phase $\phi - 1$. By Lemmas~\ref{lem:excess} and~\ref{lem:ell-reduction}, the number of new packets arriving whose paths contain $i$ during phase $\phi - 1$ satisfies
  \[
  N_{[t(\phi - 1)+1, t(\phi)]}(i) \leq \xi^{t(\phi)}(i) - \xi^{t(\phi-1)+1}(i) + \ell \cdot \rho.
  \]
  Therefore, by the inductive hypothesis, we have
  \[
  B^{t(\phi)}(i)
  ~\leq~ B^{t((\phi - 1)+)}(i) + N_{[t(\phi - 1)+1, t(\phi)]}(i)
  ~\leq~ \xi^{t(\phi)}(i) + \ell \cdot \rho 
  ~\leq~ \xi^{t(\phi)}(i) + 1.
  \]
  The inequality $B^{t(\phi+)}(i) \leq \max\set{B^{t(\phi)}(i) - 1, 0}$ follows from Lemma~\ref{lem:hpts-bad-decrease}.
\end{proof}

\section{Lower Bound}
\label{sec:lb}
We %prove Thm.~\ref{thmLB}, which shows 
show now that the upper bound achieved by $\HPTS$ is optimal, up to a 
 factor of $O(\rho^{-2})$. 

\begin{thm}\label{thmLB}
  For any $\rho > 1/(\ell+1)$, there exists a $(\rho,1)$-bounded adversary such that for all forwarding protocols, the buffer requirement is $\Omega(\frac{(\ell+1)\rho-1}{2\ell} n^{1/\ell})$.
\end{thm}

%Our proof  generalizes the construction 
%in~\cite{DBLP:conf/podc/Patt-ShamirR17}, which gives the same lower 
%bound only for $\rho > 1/2$. 
Assume that $n$ is of the form $n = (\ell + 1) m^\ell$ for some $m, \ell \in \N$, $\ell \geq 2$. We define an injection pattern $A$ consisting of $m^\ell$ phases each of length $m$. During each phase, $A$ injects packets with $\ell + 1$ non-overlapping routes, for a total of $\rho m (\ell + 1)$ packets. The right-most injection site, denoted $F(t)$, decreases by some predetermined amount at the end of each phase, and routes for the next phase are adjusted correspondingly. The routes are constructed such that by the time \emph{any} packet reaches its destination, it must be to the right of the current value of $F(t)$. Thm.~\ref{thmLB} follows by establishing the 
following dichotomy: when $F(t)$ is decreased at the end of each phase, 
either the average load of the interval $[F(t+1), F(t)]$ is 
$\Omega(m)$, or the average load behind $F(t+1)$ increased during the 
round. Thus, we show that at the end of the execution, either some 
interval $[F(t+1), F(t)]$ had a large average load at the end of some 
phase, or the average load in $[0, F(t)]$ is sufficiently large.

%We say that a packet becomes \dft{stale} in round $t$ if 
%$P(t) \leq F(t)$ and $P(t+1) > F(t+1)$. Thus, $P$ can only become stale 
%in round $t$ if either $P$ is forwarded by $F(t)$ in round $t$ (in 
%which case we call $P$ \dft{$\alpha$-stale}), or if $t$ is the last 
%round of a phase and $F(t+1)$ is moved to the left of $P$ (in which 
%case we call $P$ \dft{$\beta$-stale}). During a phase, $F(t)$ can only 
%forward $m$ packets, so the number of $\alpha$-stale packets created 
%during a phase is at most $m$. However, $A$ injects a total of $\rho 
%(\ell + 1) m$ new packets 
%during each phase.  

We express each round number $t \geq 0$ in as its base-$m$ 
representation $t_\ell t_{\ell-1}\ncdots t_0$ where $t = 
\sum_{i=0}^\ell t_i m^i$ with each $t_i \in \Seq{m}$. For each round $t 
\geq 0$, the phase containing $t$ is called the \dft{$t_\ell \ncdots 
t_1$-phase}. Equivalently, the $t_\ell \ncdots t_1$-phase consists of 
the $m$ rounds %in the interval $\left[\left(\sum_{k=1}^\ell t_k 
%m^k\right), \left(\sum_{k=1}^\ell t_k m^k\right)+m-1\right]$.
starting with $\sum_{k=1}^\ell t_k m^k$.

In each $t_\ell \ncdots t_1$-phase, the packet injection sites are 
defined based on the values of $t_\ell,\ldots,t_1$ as follows. For each 
$t_\ell \ncdots t_1$-phase, and $i \in [\ell]$ we define the $i\th$ 
site 
$$
v_i(t_\ell \ncdots t_1) = \sum_{k=i}^{\ell} \left( (k+1)m^k - 
(t_k+1)k m^{k-1} \right)~.
$$

\paragraph{The Injection Pattern} During the $t_\ell \ncdots t_1$-phase, $A$ injects packets as follows:
\begin{compactitem}
\item inject $\rho m$ packets into buffer $v_{1}(t_\ell \ncdots t_1)$ with destination $n$;
\item for each $k = 2,\ldots,\ell$, inject $\rho m$ packets into buffer $v_{k}(t_\ell \ncdots t_1)$ with destination $v_{k-1}(t_\ell \ncdots t_1)$;
\item inject $\rho m$ packets into buffer 0 with destination $v_{\ell}(t_\ell \ncdots t_1)$.
\end{compactitem}
We define the \dft{type} of a packet by its injection site: for 
each $k \in \set{1,\ldots,\ell}$, packets injected at buffer 
$v_{k}(t_\ell \ncdots t_1)$ are \dft{type-$k$} packets and 
those injected at buffer 0 are \dft{type-$(\ell+1)$} packets.

For any round $t$, we denote $F(t) = v_{1}(t_\ell \ncdots t_1)$. That is, $F(t)$ is the site of all injections of type-1 packets (which coincides with destination of all injected type-2 packets) during the $t_\ell \ncdots t_1$-phase. For any packet $P$, let $P(t)$ denote the buffer in which packet $P$ is stored in round $t$. We assign each packet a status depending on the relative values of $P(t)$ and $F(t)$: we say that $P$ is \dft{fresh} in round $t$ if $P(t) \leq F(t)$. Otherwise, if $P(t) > F(t)$, then $P$ is \dft{stale}. Note that a packet is fresh when it is injected in a round $t$, since $P(t)$ is either 0 or $v_{1}(t_\ell \ncdots t_1) = F(t)$. We say that $P$ \dft{becomes stale at the end of round $t$} if $P$ is fresh in round $t$ and is stale in round $t+1$. For any $t_\ell \in \{0,\ldots,m-1\}$, let $f(t_\ell)$ denote the number of fresh packets in the network at the beginning of the $t_\ell 0 \ncdots 0$-phase.  The following lemma gives conditions under which a packet can become stale.

\begin{lem}\label{lem:staletypes}
  A packet P can only become stale at the end of a round $t$ if either 
  (1) $P(t) = F(t)$ and $P$ is forwarded in round $t$; or (2) $F(t+1) < 
  F(t)$ and $P(t+1) \in [F(t+1)+1,F(t)]$.
\end{lem}
%\begin{proof}
%	As $F(t)$ is non-increasing, there are two possible cases:
%	\begin{itemize}
%		\item Suppose that $F(t+1) = F(t)$. If $P$ becomes stale at the end of round $t$, it means that $P$ is fresh in round $t$ and stale in round $t+1$. In particular, $P(t) \leq F(t) = F(t+1) < P(t+1)$, which implies that $P$ is forwarded in round $t$. As $P$ can only be forwarded once in round $t$, it follows that $P(t+1) = F(t) + 1$, so $P(t) = F(t)$.
%		\item Suppose that $F(t+1) < F(t)$. As $P$ is stale in round $t+1$, we have $P(t+1) > F(t+1)$, so $P(t+1) \geq F(t+1)+1$. 
%		
%		Next, as $P$ is fresh in round $t$, we have $P(t) \leq F(t)$. Assuming condition (1) doesn't hold, there are two cases to consider. In the case where $P(t) < F(t)$, since $P$ can only be forwarded once in round $t$, it follows that $P(t+1) \leq F(t)$. In the case where $P(t) = F(t)$ and $P$ isn't forwarded in round $t$, then $P(t+1)=F(t)$.
%	\end{itemize}
%\end{proof}

As packets can become stale for different reasons, we say that a packet is \dft{$\alpha$-stale} if it became stale due to condition~1 of Lemma~\ref{lem:staletypes}, and $P$ is \dft{$\beta$-stale} if it became stale due to condition~2. The following lemma shows that no packet is delivered while it is fresh. The idea is that, in the time it takes for a packet to be delivered, $F(t)$ moves to the left of the packet's destination.

\begin{lem}\label{lem:FreshUndelivered}
  If $P$ has reached its destination in round $t'$, then $P$ is stale in round $t'$.
\end{lem}

We now give a lower bound on the number of fresh packets in the network at any given time. To do so, we give an upper bound on the number of packets that go stale in each $(t_\ell \ncdots t_1)$-phase.

\begin{lem}\label{lem:boundstale}
%  Consider any $t \in \set{0,\ldots,m^{\ell-1}-1}$.
Over any interval of $\tau \geq 0$ rounds, the number of packets that become $\alpha$-stale is at most $\tau$.
Let $t\in\Seq{m^\ell}$. If $t$ is not the last round of the $t_\ell 
\ncdots t_1$-phase, then no packets become $\beta$-stale at the end of 
round $t$. If $t$ is the last round of the $t_\ell\ncdots t_1$-phase, 
and $k$ is the smallest integer in $\Seq{\ell}$ such 
that $t_{k+1} < m-1$, then the number of packets that become 
$\beta$-stale at the end of round $t$ is 
$L^{t+1}\left(\left[v_{1}(t_\ell \ncdots 
t_1)-\frac{m(m^{k}-1)}{m-1}~,~~v_{1}(t_\ell \ncdots t_1)\right]\right)~.
$
\end{lem}

\begin{lem}\label{lem:scenarios}
  Consider any fixed $t'_\ell \in \set{0,\ldots,m-2}$. At least one of the following scenarios occurs:
  \begin{compactenum}
  \item There exists a $k \in \set{0,\ldots,\ell-1}$ such that $t_{k+1} < m-1$ and at least $((\ell+1)\rho-1)m^{k+1}/2\ell$ packets become $\beta$-stale at the end of the $t'_{\ell}t_{\ell-1}\ncdots t_{k+1}(m-1)\ncdots (m-1)$-phase.
  \item $f(t'_\ell+1) \geq f(t'_\ell) + ((\ell+1)\rho-1)m^\ell/2$.		
  \end{compactenum}
\end{lem}

\begin{proof}[Proof of Thm.~\ref{thmLB}]
  First, suppose that the first scenario of Lemma~\ref{lem:scenarios} is satisfied, i.e.,  there exists a $t_\ell \in \set{0,\ldots,m-2}$ and an $k \in \set{0,\ldots,\ell-1}$ such that $t_{k+1} < m-1$ and at least $((\ell+1)\rho-1)m^{k+1}/2\ell$ packets become $\beta$-stale at the end of the $t_\ell\ncdots t_{k+1}(m-1)\ncdots (m-1)$-phase. Then, by Lemma~\ref{lem:boundstale}, $L^{t+1}([v_{1}(t_\ell \ncdots t_1)-m(m^{k}-1)/(m-1),v_{1}(t_\ell\ncdots t_1)]) \geq ((\ell+1)\rho-1)m^{k+1}/2k$. As $[v_{1}(t_\ell\ncdots t_1)-m(m^{k}-1)/(m-1),v_{1}(t_\ell\ncdots t_1)]$ consists of $O(m^k)$ buffers, the average load of a buffer in this range is at least $\Omega(((\ell+1)\rho-1)m/2\ell)$. Hence, there is at least one buffer with load at least $\Omega(((\ell+1)\rho-1)m/2\ell)$, giving the desired conclusion.
	
  Otherwise, suppose that the first scenario from Lemma~\ref{lem:scenarios} is never satisfied in any $t_\ell\ncdots t_1$-phase with $t_\ell \in \set{0,\ldots,m-2}$. Then, by Lemma~\ref{lem:scenarios}, the second scenario occurs for every $t_\ell \in \set{0,\ldots,m-2}$. Therefore, $f(m-1) \geq (m-1)((\ell+1)\rho-1)m^\ell/2$. By Lemma~\ref{lem:FreshUndelivered}, the total load over all buffers in the first round $t$ of the $(m-1)0\ncdots 0$-phase is bounded below by the number of fresh packets, i.e., $L^t([n]) \geq f(m-1) \geq (m-1)((\ell+1)\rho-1)m^\ell/2$. Therefore, the average load in the network in round $t$ is at least $(m-1)((\ell+1)\rho-1)m^\ell/2n = (m-1)((\ell+1)\rho-1)m^\ell/2(\ell+1)m^\ell = (m-1)((\ell+1)\rho-1)/2(\ell+1) \in \Omega(\frac{(\ell+1)\rho-1}{2\ell} n^{1/\ell})$, and the conclusion of the theorem follows.
\end{proof}

\newpage
\pagenumbering{roman}

\bibliographystyle{abbrv}
{\small \bibliography{intro}}

\newpage
\section*{APPENDIX}
\appendix

\section{Proofs for Section \ref{sec:preliminaries}}
\label{app:preliminaries}
\begin{proof}[Proof of Lemma~\ref{lem:excess}]
  The first claim follows from the definition of $(\rho, \sigma)$ bound: For all $T = [s, t]$ we have $N_T(v) \leq \rho \abs{T} + \sigma$, hence $N_T(v) - \rho \abs{T} \leq \sigma$. For the second claim, suppose $\nu$ packets are injected in round $t$ whose trajectories contain $v$, and let $T' = [s, t - 1]$ be an interval achieving the maximum in the definition of $\xi^{t-1}(v)$. Then we have
  \begin{align*}
    \xi^t(v) &\geq N_{[s, t]}(v) - \rho (t - s + 1)\\ &= N_{[s, t-1]}(v) + \nu - \rho((t - 1) - s + 1) - \rho\\ &\geq \xi^{t-1}(v) + \nu - \rho
  \end{align*}
  Rearranging this expression gives $\nu \leq \xi^t(v) - \xi^{t-1}(v) + \rho$, as desired.
\end{proof}

\begin{proof}[Proof of Lemma~\ref{lem:ell-reduction}]
  Fix an interval of time $T = [s, t]$ and buffer $v \in V$.  From the definition of $A_\ell$, we have
  \begin{align*}
    N_T(v) &= \abs{\set{(u, i_P, w_P) \in A \sucht \ell \cdot s + 1 \leq u \leq \ell \cdot t}}\\ &\leq \rho \cdot \ell (t - s + 1) + \sigma.
  \end{align*}
  The inequality holds because $A$ is $(\rho, \sigma)$-bounded.  Therefore, $A_\ell$ is $(\ell\cdot\rho, \sigma)$-bounded, as desired.
\end{proof}

\section{Additional Material for Section \ref{sec:basic-alg}}
\label{app:basic-alg}
\subsection{Proofs}
\begin{proof}[Proof of Prop.~\ref{prop:pts}]
	For each time $t$, let $B^t(w-1)$ denote the number of bad packets 
	in the network after the $t\th$ injection step. That is
	\[
	B^t(w-1) = \sum_{i \in \Seq{n}} \max \set{L^t(i) - 1, 0}.
	\]
	Similarly, let $B^{t+}(w-1)$ be the number of bad packets after the 
	$t\th$ forwarding step. We will argue by induction that for all $t$
	\begin{equation}
	\label{eqn:pts-bad}
	B^t(w-1) \leq \xi^t(w-1) + 1 \quad\text{and}\quad B^{t+}(w-1) \leq 
	\xi^t(w-1).
	\end{equation}
	Since the maximum load of any buffer can be at most $B(t) + 1$, the 
	proposition follows.
	
	For the base case, $t = 0$, (\ref{eqn:pts-bad}) trivially holds as 
	all quantities are zero. For the inductive step, 
	suppose~(\ref{eqn:pts-bad}) holds for round $t - 1$.  Let $\nu(t)$ 
	denote the number of new packets injected during the $t\th$ 
	injection step. By part~2 of Lemma~\ref{lem:excess}, $\nu(t) \leq 
	\xi_{w-1}(t) - \xi_{w-1}(t - 1) + 1$. Therefore
	\[
	B^t(w-1) \leq B^{t-1}(w-1) + \nu(t) \leq \xi^t(w-1) + (\xi^t(w-1) - 
	\xi^{t - 1}(w-1) + 1) = \xi^t(w-1) + 1.
	\]
	Thus the first expression in~(\ref{eqn:pts-bad}) holds, as desired. 
	To show the second expression also holds, it suffices to show that 
	if $B^t(w-1) \geq 1$, then $B^{t+}(w-1) = B^t(w-1) - 1$. To this 
	end, let $L^{t+}$ denote the configuration after forwarding. Let 
	$i$ be the left-most buffer for which $\abs{L^t(i)} \geq 2$. Since 
	$i$ forwards, but $i - 1$ does not, we have $L^{t+}(i) = L^t(i) - 1 
	\geq 1$. Since every non-empty $i' > i$ forwards, and no $i'$ 
	receives more than one packet from $i' - 1$, we have for all $i' > 
	i$
	\[
	\abs{L^{t+}(i')} \leq \max{1, \abs{L^t(i')}}.
	\]
	Therefore, we compute
	\begin{align*}
	B^{t+}(w-1) &= (\abs{L^{t+}(i)} - 1) + \sum_{i' > i} \max 
	\set{\abs{L^{t+}(i')} - 1, 0}\\
	&\leq L^t(i) - 2 + \sum_{i' > i} \max\set{\abs{L^t(i')} - 1, 0}\\
	&= B^t(w-1) - 1,
	\end{align*}
	as desired.
\end{proof}

To prove Prop.~\ref{prop:ppts}, we first
% First, we 
 show that the forwarding pattern prescribed by $\PPTS$ is 
feasible, in the sense that at most one pseudo-buffer in a buffer $i$ 
forwards in each round.
 \begin{lem}
   \label{lem:ppts-feasible}
 %  Fix a configuration $L$ and consider a single step of $\PPTS$. 
%For 
 %each $k \in [d]$ let $I_k = [a_k, b_k]$ denote the (possibly empty) 
 %interval of $k$-pseudo-buffers activated during the execution. Then 
 %for 
 %all non-empty $I_k, I_{k'}$ with $k < k'$, we have $b_{k} < 
%a_{k'}$. 
 %In 
 %particular, the intervals created in a step of $\PPTS$ are 
%pair-wise 
 %disjoint, hence at most one pseudo-buffer is activated in each 
%buffer.
   Consider a single step of $\PPTS$, and $I_k$ be a the (possibly 
   empty) interval 
   of 
   $k$-pseudo-buffers activated during that step. Then for any 
distinct
   $k,k'\in\Seq{d}$ we have $I_k\cap I_{k'}=\varnothing$.
 \end{lem}

\begin{proof}[Proof of Lemma~\ref{lem:ppts-feasible}]
	Denote the endpoints of an interval $I_k$ by $[a_k,b_k]$.
	Let $i(k)$ be the value of $i$ at the beginning of iteration $k$ of 
	the loop in Lines~3--9 $\PPTS$. Note that $I_k$ is empty, unless 
	the condition of Line~4 is satisfied. In the former case, we have 
	$i(k + 1) = i(k)$, and in the latter case we have $i(k+1) = i_k < 
	i(k)$ and $I_k = [i_k, i(k) - 1]$. Therefore, the sequence $i(d), 
	i(d-1), \ldots$ is weakly decreasing. Thus, for $k < k'$, if $I_k, 
	I_{k'} \neq \varnothing$, there is some $k''$ satisfying $k < k'' 
	\leq k'$ such that
	\(
	b_k = i(k'') - 1 < i_{k''} \leq i_{k'} = a_{k'}
	\)
	which gives the desired result.
\end{proof}

%\begin{lem}
%	\label{lem:bad-decrease}
%	Let $L$ be a configuration of $\Seq{n}$ with destinations $W = 
%	\set{w_0 < w_1 < \cdots < w_{d-1}}$. Suppose that for some $k \in 
%	\Seq d$, $I_k = [a_k, b_k] \subseteq \Seq{n}$ is an interval such 
%	that $\abs{L_k(a_k)} \geq 2$ and $b_k < w_k$. Let $L^+$ be the 
%	configuration obtained by forwarding only from nonempty 
%	$k$-pseudo-buffers in $I_k$. For $i \in \Seq{n}$ let $B_k(i)$ and 
%	$B_k^+(i)$ be the $k$-badness of $i$ before and after forwarding, 
%	respectively. Similarly, $B(i)$ and $B^+(i)$ denote the badness of 
%	$i$ before and after forwarding. Then
%	\[
%	\begin{array}{lcl}
%	B_k^+(i) \leq
%	\begin{cases}
%	B_k(i) - 1 &\text{if } i \in I_k\\
%	B_k(i) &\text{if } i \notin I_k~,
%	\end{cases}
%	&
%	\qquad\text{and similarly,}\qquad&
%	B^+(i) \leq
%	\begin{cases}
%	B(i) - 1 &\text{if } i \in I_k\\
%	B(i) &\text{if } i \notin I_k~.
%	\end{cases}
%	\end{array}
%	\]
%\end{lem}

\begin{proof}[Proof of Lemma~\ref{lem:bad-decrease}]
	First observe that since only $k$-packets are forwarded, so only 
	$k$-loads---hence $k$-badness---changes after forwarding. Thus, the 
	second claim follows immediately from the first. 
	
	For a buffer $i$, let $\beta_k(i)$ and $\beta_k^+(i)$ denote the 
	number of $k$-bad packets stored in $i$ before and after 
	forwarding, respectively. We consider 4 cases separately:
	\begin{compactitem}
		\item $i < a_k$ or $i > b_k + 1$. In this case, $i$ neither 
		sends nor receives a packet so that $\beta_k^+(i) = \beta_k(i)$.
		\item $i = a_k$. In this case, $L_k(a_k) \geq 2$. Since $a_k$ 
		sends but does not receive a packet, we have $\beta_k^+(a_k) = 
		\beta_k(a_k) - 1$.
		\item $a_k + 1 \leq i \leq b_k$. Observe that $i$ receives at 
		most one packet, and if $L_k(i) \geq 1$, $i$ forwards one 
		packet. Thus $L_k^+(i) \leq \max\set{L_k(i), 1}$, hence 
		$\beta_k^+(i) \leq \beta_k(i)$.
		\item $i = b_k + 1$. In this case, $i$ receives at most one 
		packet, hence $\beta_k^+(i) \leq \beta_k(i) + 1$.
	\end{compactitem}
	The lemma follows from the four cases above, using the definitions
	of $B_k(i)$ and $B_k^+(i)$.
	%  \[
	%  B_k(i) = \sum_{i' \leq i} \beta_k(i') \quad\text{and}\quad 
	%B_k^+(i) 
	%= \sum_{i' \leq i} \beta_k^+(i').
	%  \]
\end{proof}

\begin{proof}[Proof of Proposition~\ref{prop:ppts}]
	For each round $t > 0$ we, use $t+$ to denote the time after the 
	forwarding step of round $t$. For example, $B^{t+}(i)$ is the 
	badness 
	of $i$ \emph{after} the $t\th$ forwarding step (whereas $B^{t}(i)$ 
	is 
	the badness after the injection step, before forwarding). We will 
	show that for all rounds $t$ and buffers $i \in \Seq{n}$, we have
	\begin{equation}
	\label{eqn:bad-bound}
	B^t(i) \leq \xi^t(i) + 1 \quad\text{and}\quad B^{t+}(i) \leq 
	\xi^t(i).
	\end{equation}
	To see that the proposition follows from~(\ref{eqn:bad-bound}), we 
	compute
	%  \begin{align*}
	%    L^t(i) &= \sum_{k \in [d]} L_k^t(i)\\
	%    &\leq \sum_{k \in [d]} (1 + \beta_k^t(i))\\
	%    &\leq d + \sum_{k \in [d]} \beta_k^t(i)\\
	%    &\leq d + \sum_{k \in [d]} B_k^t(i)\\
	%    &= d + B^t(i).
	%  \end{align*}
	\begin{align*}
	|L^t(i)| 
	~= \sum_{k \in \Seq{d}}\! |L_k^t(i)|
	\leq \sum_{k \in \Seq{d}}\! (1 + \beta_k^t(i))
	~\leq~ d +\!\! \sum_{k \in \Seq{d}}\! \beta_k^t(i)
	~\leq~ d +\!\! \sum_{k \in \Seq{d}}\! B_k^t(i)
	=~ d + B^t(i)~.
	\end{align*}
	Therefore, from~(\ref{eqn:bad-bound}), we obtain $|L^t(i)| \leq d + 
	1 + 
	\xi^t(i)$. By Claim~1 of Lemma~\ref{lem:excess}, $\xi^t(i) \leq 
	\sigma$, so that $L^t(i) \leq d + 1 + \sigma$, as desired.
	
	In order to prove~(\ref{eqn:bad-bound}), we argue by induction on 
	$t$. The proof for $t = 0$ is trivial, as no packets have yet been 
	injected. Suppose that~(\ref{eqn:bad-bound}) holds for every $i \in 
	\Seq{n}$ at time $t - 1$. Let $N_{\{t\}}(i)$ denote the number of 
	packets injected in round $t$ whose paths contain $i$. We compute
	\begin{align*}
	B^t(i) &\leq B^{(t-1)+}(i) + N_{\{t\}}(i)\\
	&\leq \xi^{t-1}(i) + N_{\{t\}}(i) &\text{(inductive hypothesis)}\\
	&\leq \xi^{t-1}(i) + (\xi^t(i) - \xi^{t-1}(i) + 1) 
	&\text{(Lemma~\ref{lem:excess})}\\
	&= \xi^t(i) + 1~.
	\end{align*}
	Therefore, the first equation of~(\ref{eqn:bad-bound}) holds.
	
	To prove the second equation in~(\ref{eqn:bad-bound}), it suffices 
	to 
	show that if $B^t(i) > 0$, then after the $t\th$ forwarding step we 
	have $B^{t+}(i) = B^t(i) - 1$. For each $k \in \Seq{d}$, let $I_k$ 
	be 
	the 
	(possibly empty) interval of $k$-pseudo-buffers activated in Line~6 
	of $\PPTS$. We have the following.
	\begin{claim}
		\label{claim}
		Let $\iota \in \Seq{n}$. If $B^t(\iota) > 0$, then 
		there exists $\kappa \in\Seq d$ such that $\iota \in I_\kappa$.
	\end{claim}
	\begin{proof}[Proof of Claim \ref{claim}]
		Let $\kappa$ be the maximum index such that $B_\kappa^t(\iota) 
		> 
		0$. Thus, for all $k > \kappa$, there is no $k$-bad packet in 
		any 
		buffer $i \leq \iota$. Therefore, in the corresponding 
		iteration $k 
		= \kappa$ of $\PPTS$, we have $i_k \leq \iota$ in Line~5. Since 
		$i$ 
		is only ever assigned to a buffer containing a bad packet 
		(Line~7 
		of $\PPTS$), during the $k\th$ iteration we must have had 
		$\iota < 
		i$ in Line~6 when $I_k = [i_k, i - 1]$ is activated. Thus, 
		$\iota 
		\in I_k$, as desired.
	\end{proof}

	Finally, suppose $i$ satisfies $B^t(i) > 0$. Applying the claim, 
	let 
	$I_k = [a_k, b_k]$ be the activated interval containing $i$. 
	Further, 
	by the definition of $i_k (= a_k)$ in Line~5 of $\PPTS$, we have 
	$\beta^t_k(a_k) > 0$. Therefore, by Lemma~\ref{lem:bad-decrease}, 
	we 
	have $B_k^{t+}(i) \leq B_k^{t}(i) - 1$. For all $k' \neq k$, 
	Lemma~\ref{lem:bad-decrease} gives $B_{k'}^{t+}(i) \leq 
	B_{k'}^{t}$. 
	Summing over all destinations gives
	\[
	B^{t+}(i) 
	~= \sum_{k' \in\Seq{d}} B_{k'}^{t+}(i) 
	~=~ B_{k}^{t+}(i) + 
	\sum_{k' \neq k} B_{k'}^{t+}(i) \leq B_k^t(i) - 1 + \sum_{k' \neq 
		k} 
	B_k^t(i) 
	~\leq~ B_k^t(i) - 1,
	\]
	as desired.
\end{proof}

\subsection{Generalization to directed trees}
\label{app:trees}
Both  $\PTS$ and $\PPTS$ algorithms generalize to directed trees 
topologies. In this section we explain how.

 A \dft{directed tree} $G = (V, E)$ is a rooted tree, where 
all edges are directed toward the root $r \in V$. We assume that the 
trajectories of all packets are directed paths in $G$.

The orientation of edges $e \in E$ toward $r$ induces a partial order 
$\prec$ on $V$:  for $u, v \in V$, we have 
$u \prec v$ if and only if $v$ is on the unique path from $u$ to $r$. 
Thus, leaves are minimal with respect to $\prec$,
% (in the sense that if $v$ is a leaf, there is no $u$ satisfying $u 
%\prec v$), 
and the root 
$r$ is maximal. By extension, given a configuration $L$ of $G$, $\prec$ 
induces a partial order on packets, where $P \prec P'$ if $P \in L(u)$, 
$P' \in L(v)$, and $u \prec v$.

\medskip
\paragraph{Single destination case.} 
Suppose first that all packets 
have destination $r$, the root of $G$.  Again, we say that a packet $P 
\in L^t(v)$ is bad if it stored at position $p \geq 2$, and $L^t(v)$ is 
bad if $\abs{L^t(v)} \geq 2$. The generalized $\PTS$ in this case 
selects a maximal set of ``farthest'' bad buffers from the root, and 
activates the union of paths from these buffers to the root. 
Formally, we have the following.

\begin{dfn}
	\label{dfn:low-antichain}
	Let $\calB^t \subseteq V$ be the set of bad buffers at time $t$.  
	$\calP \subseteq \calB^t$ is a \dft{low-antichain} if
	\begin{inparaenum}[(1)]
		\item %$\calP$ is an antichain with respect to $\prec$; that 
		%is, for
		For all $u, v \in \calP$, $u \nprec v$ and $v \nprec u$; and
		\item For every $v \in \calB$ there exists $u \in \calP$ such 
		that $u 
		\preceq v$.
	\end{inparaenum}
\end{dfn}

Given $\calB^t$, there is a unique low-antichain
%, $\min(\calB^t)$, 
%defined by
\(
\min(\calB^t) = \set{u \in \calB^t \sucht \forall v \in \calB^t, v 
\nprec u}.
\)
Our generalization of $\PTS$ activates buffers as follows: Let $\calP = 
\min(\calB^t)$. Then define $\calA$ to be
\(
\calA = \bigcup_{u \in \calP} \Path(u, r),
\)
where $\Path(u, v)$ returns the set of nodes on the unique path from 
$u$ to $v$ in $G$. More concisely, $\PTS$ activates every buffer $v \in 
V$ such that there exists a bad buffer $u$ with $u \preceq v$. We state 
the performance of this variant of $\PTS$ in 
the following.

\begin{prop}
	\label{prop:tree-pts}
	Suppose $G$ is a directed tree, $A$ is a $(\rho, \sigma)$-bounded 
	adversary with $\rho \leq 1$, and all packets $P \in A$ have 
	destination $r$. Then then maximum buffers size is at most $2 + 
	\sigma$.
\end{prop}

Proposition~\ref{prop:tree-pts} is proven analogously to 
Proposition~\ref{prop:pts}. 
Specifically, we bound the number of bad packets ``behind'' any node 
$v \in V$ by $\xi^t(v) + 1$ during any step of the algorithm. To this 
end, we generalize the notion of badness of a node to directed trees.

\begin{dfn}
  Let $v \in V$. Let $G_v = (U_v, E_v)$ denote the (directed) subtree 
of $G$ rooted at $v$, so that $U_v = \set{u \in V \sucht u \preceq 
v}$. 
For each $v \in V$ we define $\beta^t(v) = \max\set{\abs{L^t(v)}-1, 
0}$ 
to be the number of bad packets in $v$. The \dft{badness} of $v$ is
  \(
  B^t(v) = \sum_{u \in U_v} \beta^t(u).
  \)
  That is, $B^t(v)$ is the number of bad packets whose paths contain 
$v$.
\end{dfn}

\begin{lem}
  \label{lem:tree-bad-decrease}
  Let $L^t$ be a configuration and $\calA$ the set of activated nodes 
described above. For any $v \in V$, let $B^{t+}(v)$ denote the badness 
of $v$ immediately after forwarding. Then
  \(
  B^{t+}(v) \leq \max\set{B^t(v) - 1, 0}.
  \)
\end{lem}
\begin{wrapfigure}{r}{.45\textwidth}
	\begin{minipage}{.45\textwidth}\vspace{-8mm}
		\begin{algorithm}[H]\small
			\caption{$\PPTS(G, W)$ %Step 
				(tree variant)}
			\label{alg:tree-pp2s}
			\begin{algorithmic}[1]
				\STATE $\calA \leftarrow \varnothing$
				\FOR{$k\gets d-1$ \Downto $0$}
				\STATE $\calP \leftarrow \min(\calB^t_k)$
				\STATE $\calA_k \leftarrow \paren{\bigcup_{u \in \calP} 
					\Path(u, w_k)} \setminus \calA$
				%	\hfill\COMMENT{activate $k$-pseudo-buffers}
				\STATE $\calA \leftarrow \calA \cup \calA_k$ 
				\ENDFOR
				\STATE every active (nonempty) pseudo-buffer in $\calA$ 
				forwards a packet
			\end{algorithmic}
		\end{algorithm}
	\end{minipage}
\end{wrapfigure}

Given Lemma~\ref{lem:tree-bad-decrease}, 
Prop.~\ref{prop:tree-pts} follows from Lemma~\ref{lem:excess} as 
in the proof of Prop.~\ref{prop:pts}.

%\paragraph{Multiple destination case} We now 
Next, we generalize Algorithm $\PPTS$   to directed trees. 
As before let $W = \set{w_1, w_2, \ldots, w_d}$ be a set of possible 
destinations of packets. We assume that the sequence $w_1, w_2, \ldots, 
w_d$ is a topological sorting of $W$ so that $w_i \prec w_j$ implies 
that $i < j$. As in Section~\ref{sec:pp2s}, we partition in $L^t(v)$ 
into pseudo-buffers $\set{L_k^t(v) \sucht k \in \Seq{d}}$, where 
$L_k^t(v)$ stores packets destined for $w_k$. Badness of pseudo-buffers 
is defined analogously to before:
\[
\beta_k^t(v) = \max\set{\abs{L_k^t(v)} - 1, 0}, \qquad B_k^t(v) = 
\sum_{u \in U_v} \beta_k^t(u), \qquad B^t(v) = \sum_{k \in \Seq{d}} 
B_k^t(v).
\]
We denote $\calB_k^t = \set{v \in V \sucht \beta^t_k(v) \geq 1}$ to be 
the set of nodes containing $k$-bad pseudo-buffers.

The $\PPTS$ algorithm for directed trees 
%activates $k$-pseudo-buffers in descending order as 
is described in Algorithm~\ref{alg:tree-pp2s}. In 
order to quantify the performance of Algorithm~\ref{alg:tree-pp2s}, we 
define the \dft{destination depth} $d' = d'(G, W)$ of $G$ with 
destinations $W$ to be the maximal number of destinations on
a leaf-root path, or, equivalently, the length of the longest chain
(w.r.t.\ $\prec$) of 
destinations. 
%in $W$. 
%That is,
%\(
%d' = \max\set{s \sucht \exists  w_1, \ldots, w_s\in W, w_1 %\prec w_2 
%	%\prec 
%	\cdots \prec w_s}.
%\)
Note that $d' \leq \min\set{d, D}$ where $D$ is the depth of $G$.

%\begin{prop}[Proposition \ref{prop:tree-ppts}]
%%	\label{prop:tree-ppts}
%	Let $A$ be any $(\rho, \sigma)$-bounded adversary such that all 
%	packets have destinations in $W$. Then the maximum buffer usage of 
%	$\PPTS$ is at most $1 + d' + \sigma$.
%\end{prop}

The proof of Proposition~\ref{prop:tree-ppts} is analogous to the proofs 
of Propositions~\ref{prop:ppts} and~\ref{prop:tree-pts}. 
Specifically, an argument analogous to 
Lemma~\ref{lem:tree-bad-decrease} shows that forwarding $\calA$ 
decreases the badness of each buffer. Combined with 
Lemma~\ref{lem:excess}, this implies that the badness of each node $v$ 
is bounded by $\xi^t(v) + 1 \leq \sigma + 1$. Since each $v$ can 
contain at most $d'$ non-bad packets, Proposition~\ref{prop:tree-ppts} 
follows from this bound on the badness of each node.

\section{Proofs for Section \ref{sec:hierarchical-alg}}
\label{app:hierarchical-alg}

\begin{proof}[Proof of Lemma~\ref{lem:hpts-feasible}]
  Consider the state of $\calA$ immediately after the calls $\FormPaths$ in Lines~\ref{l:ppts1}--\ref{l:ppts2}. The intervals $I_{\lambda, r}$ are disjoint, and on each interval, $I_{\lambda, r}$ the level-$\lambda$ pseudo-buffers are activated in the same pattern as a call to $\PPTS$. Thus the feasibility of $\calA$ after Line~\ref{l:ppts2} follows from the observation that the activation pattern for $\PPTS$ is feasible. $\calA$ remains feasible after each call to $\ActivatePreBad$, because only previously inactive buffers are activated in Line~\ref{l:inactive}.
\end{proof}

In order to prove Lemma~\ref{lem:hpts-bad-decrease}, we will apply the 
following result.

\begin{claim}
  \label{clm:hpts-bad-decrease}
  Let $t = (\varphi - 1) \ell$ and $s = t + \lambda$ for some $\lambda \in \Seq{\ell}$. Then for all $i \in \Seq{n}$ the following hold:
  \begin{compactenum}
  \item For all $j < \lambda$, $B^{s+}_j(i) = B^s_j(i)$.
  \item $B^{s+}_\lambda(i) \leq \max\set{B^s_\lambda(i) - 1, 0}$.
  \item $B^{s+}(i) \leq B^s(i)$.
  \end{compactenum}
\end{claim}
\begin{proof}
  Item~1 holds because $\HPTS$ can only activate buffers at levels $j' \leq \lambda$ in round $t + \lambda$ and no packet ever moves from level $j'$ to $j$ with $j > j'$. Thus states of level-$j$ pseudo buffers are unchanged after forwarding for $j > \lambda$.  For Item~2, observe that for each interval $I \in \calI_\lambda$ the paths activated in $I$ are precisely the same pseudo-buffers as are activated by $\PPTS$ (where the intermediate destinations play the role of destinations in the $\PPTS$ algorithm). Thus, Item~2 follows from Lemma~\ref{lem:bad-decrease}.
  
  For Item~3, consider a level-$j$ segment $A$ added to $\calA$ in a call to $\ActivatePreBad$, and let $I_j = [a_j, b_j] \in \calI$ be the level-$j$ interval containing $A$. Thus $A = [a_j, i_j]$ for some $i_j < b_j$. Let $P$ be the unique pre-bad packet for $a_j$, and let $k$ be the index such that $w_{P, j}$ is $I_j$'s $k\th$ destination. Observe that one of the following events occurred:
  \begin{compactenum}[(a)]
  \item $i_j = w_{P, j} - 1$ and activating $A$ creates a new pre-bad packet $P'$ in buffer $L_{j, k}^{s}(i_j)$.
  \item $i_j = w_{P, j} - 1$ but activating $A$ does not create another pre-bad packet.
  \item $i_j < w_{P, j} - 1$ and every $i \in I_j$ with $i > i_j$ is 
  active at at level $\lambda$.\footnote{The second conclusion holds 
  because intervals activated in $\FormPaths$ always terminate at 
  intermediate destinations for level $\lambda$.  Since $j < \lambda$, 
  $I_j$ is contained in the interval between consecutive intermediate 
  level $\lambda$ intermediate destinations, if some $i \in I_j$ is 
  active at level $\lambda$, all buffers $i' \in I_j$ with $i' \geq i$ 
  are active at level $\lambda$.}
  \end{compactenum}
  From items~(a)--(c) above, each pre-bad $P$ created in the call to $\FormPaths$ induces a (possibly empty) sequence $A_1, A_2, \ldots$ of consecutive segments at decreasing levels, where each segment $A_{i + 1}$ is activated because event~(a) occurred in the activation of $A_i$. Conversely, every pre-bad packet $P'$ at a level $j < \lambda$ is induced from a unique pre-bad packet $P$ at level $\lambda$. To prove part~3 of the claim, let $i \in I_j$ for some $j < \lambda$. We consider three (exhaustive) cases separately.
  \begin{compactdesc}
  \item[Case 1:] $i \in A_\alpha$ where $A_\alpha$ is the unique activated containing $i$.
  \item[Case 2:] $i$ is active at level $\lambda$, and an active segment $A_\alpha \subseteq I_j$ terminates as in case~(c) above (necessarily with $i_j < i$).
  \item[Case 3:] $i$ is not active.
  \end{compactdesc}
  In Case~1, suppose the segment $A_\alpha$ is at level $j$. First observe that for $j' \neq j$, we have $B_{j'}^{s+}(i) = B_{j'}^s(i)$ because no buffer $i' \leq i$ at level $j'$ containing packets with intermediate destination $w \geq i$ is active. As for the level-$j$ badness of $i$, we have $B_{j}^{s+}(i) \leq B_{j}^{s}(i)$. To see this, let $I_j = [a_j, b_j]$ be the level-$j$ interval containing $i$, so that $a_j$ receives a pre-bad packet $P$ during forwarding. Let $k$ be the index such that $w_{P, j}$ is $I_j$'s $k\th$ intermediate destination. If $\abs{L_{j, k}^{s}(a_j)} \geq 2$, then forwarding $A_\alpha$ strictly decreases $B_{j, k}(i)$ by Lemma~\ref{lem:bad-decrease}, while $B_{j}(i)$ increases by one as a result of $P$ being forwarded. Similarly, if $\abs{L_{j, k}^{s}(a_j)} = 1$, forwarding $A_\alpha$ does not increase $B_{j, k}(i)$. In this case, $P$ being forwarded to $a_j$ does not increase $B_{j}(i)$, because $L_{j, k}^s(a_j)$'s sole occupant is forwarded.
  
  In Case~2, for all $j' \neq \lambda, j$, we have $B_{j'}^{s+}(i) = B_{j'}^{s}(i)$. $B_{\lambda}^{s+}(i) \leq B_{\lambda}^s(i) - 1$ by Lemma~\ref{lem:bad-decrease}. We argue that $B_{j}^{s+}(i) \leq B_{j}^s(i) + 1$. To see this first consider the case where $\abs{L_{j, k}(a_j)} \geq 2$. In this case, $\beta_{j, k}^{s+}(a_j) = \beta_{j, k}^s(a_j)$ because $a_j$ forwards one bad packet, and receives one bad packet. For $i' \in [a_j + 1, i_j]$, $\beta_{j, k}(i')$ does not increase because each non-empty $L_{j, k}(i')$ forwards. Finally $\beta_{j, k}(i_j + 1)$ increases by at most one as a result of $i_j + 1$ receiving a level-$j$ packet, but sending none. Since no other level-$j$ pseudo-buffers forward, in this case, $B_j^{s+}(i) \leq B_j^s(i) + 1$. If $\abs{L_{j, k}(i)} = 1$, then forwarding $P$ does not create a bad packet, but forwarding $A_\alpha$ could create a single new bad packet (in buffer $i_j + 1$), so the same conclusion holds.
  
  Finally for Case~3, no packet that is pre-bad for any level $j$ has a path crossing $i$ at level $j$ (otherwise Cases~1 or~2 would have occurred). Therefore, forwarding pre-bad packets will not increase $B_j(i)$ for any $j < \lambda$. Further $B_j^{s+}(i) = B_j^s(i)$ for all $j \geq \lambda$, as no level-$j$ buffers are active that forward packets whose paths cross $i$ at level $j$.
\end{proof}

\begin{proof}[Proof of Lemma~\ref{lem:hpts-bad-decrease}]
  Since we assume no injections are made in rounds $t + 2, t + 3, \ldots, t + \ell$, we have $B^{s+1}(i) = B^{s+}(i)$ for all $s = t+1, t+2, \ldots, t + \ell$. Applying Item~3 of Claim~\ref{clm:hpts-bad-decrease} iteratively for $t + 1, t + 2, \ldots, t + \ell$, we find that $B^{(t + \ell)+}(i) \leq B^t(i)$.  Thus, it suffices to show that if $B^{t+1}(i) \geq 1$, there is some round $s = t + \lambda$ for $\lambda \in \Seq{\ell}$ such that $B^{s+}(i) < B^s(i)$. By Claim~\ref{clm:hpts-bad-decrease}, in each round $s$ such that $B^{s}(i) > 0$, either $B^{s+}(i) \leq B^{s}(i) - 1$, or \[ B^{s+}_{\lambda + 1}(i) + B^{s+}_{\lambda + 2}(i) + \cdots + B^{s+}_{\ell}(i) = B^{s}_{\lambda + 1}(i) + B^{s}_{\lambda + 2}(i) + \cdots + B^{s}_{\ell}(i) + 1.  \] That is, either the badness of $i$ decreases, or the level-$j$ badness of $i$ for some $j < \lambda$ increases. Consider the largest $s = t + \lambda$ such that $B^s_\lambda(i) > 0$. Then the level-$j$ badness of $i$ did not increase for any $j < \lambda$, hence we have $B^{s+}(i) = B^s(i)$, as desired.
\end{proof}

\section{Proofs for Section \ref{sec:lb}}
\label{app:lb}

\begin{proof}[Proof of Lemma~\ref{lem:staletypes}]
		As $F(t)$ is non-increasing, there are two possible cases:
		\begin{itemize}
			\item Suppose that $F(t+1) = F(t)$. If $P$ becomes stale at the end of round $t$, it means that $P$ is fresh in round $t$ and stale in round $t+1$. In particular, $P(t) \leq F(t) = F(t+1) < P(t+1)$, which implies that $P$ is forwarded in round $t$. As $P$ can only be forwarded once in round $t$, it follows that $P(t+1) = F(t) + 1$, so $P(t) = F(t)$.
			\item Suppose that $F(t+1) < F(t)$. As $P$ is stale in round $t+1$, we have $P(t+1) > F(t+1)$, so $P(t+1) \geq F(t+1)+1$. 
			
			Next, as $P$ is fresh in round $t$, we have $P(t) \leq F(t)$. Assuming condition (1) doesn't hold, there are two cases to consider. In the case where $P(t) < F(t)$, since $P$ can only be forwarded once in round $t$, it follows that $P(t+1) \leq F(t)$. In the case where $P(t) = F(t)$ and $P$ isn't forwarded in round $t$, then $P(t+1)=F(t)$.
		\end{itemize}
\end{proof}

\begin{proof}[Proof of Lemma~\ref{lem:FreshUndelivered}]
	We first prove the following claim that relates the values of $F(t)$ and $F(t')$ depending on the value of $t' - t$.
		\begin{claim}\label{claim:FMoves}
			Consider any round $t$ and any round $t' > t$. If there exists $k \in \set{1,\ldots,\ell}$ such that $t'_k > t_k$ and $t'_i = t_i$ for all $i \in \set{k+1,\ldots,\ell}$, then $F(t') < v_k(t_\ell \ncdots t_1)$.
		\end{claim}
		
		By definition,
		\begin{align*}
		F(t') ~ =~ v_1(t'_\ell \ncdots t'_1)%\\
		& = v_1(t_\ell \ncdots t_{k+1}t'_k \ncdots t'_1)\\
		& = \sum_{j=1}^{k} \left[ (j+1)m^j - (t'_j+1)jm^{j-1} \right] + 
		\sum_{j=k+1}^{\ell} \left[ (j+1)m^j - (t_j+1)jm^{j-1} \right]~.
		\end{align*}
		By definition, $v_k(t_\ell \ncdots t_1) = \sum_{j=k}^{\ell} \left[ (j+1)m^j - (t_j+1)jm^{j-1} \right]$. Therefore,
		\begin{align*}
		F(t') - v_k(t_\ell \ncdots t_1) & = \sum_{j=1}^{k} \left[ (j+1)m^j - (t'_j+1)jm^{j-1} \right] - \left[(k+1)m^k - (t_k+1)km^{k-1}\right]\\
		& = \sum_{j=1}^{k-1} \left[ (j+1)m^j - (t'_j+1)jm^{j-1} \right] - \left[(t'_k+1)km^{k-1} - (t_k+1)km^{k-1}\right]\\
		& < \sum_{j=1}^{k-1} \left[ (j+1)m^j - (t'_j+1)jm^{j-1} \right] - \left[km^{k-1}\right]\\
		& \leq \sum_{j=1}^{k-1} \left[ (j+1)m^j - jm^{j-1} \right] - \left[km^{k-1}\right]\\
		& = \left[ km^{k-1} - 1\right] - \left[km^{k-1}\right]\\
		& = -1~.
		\end{align*}
		It follows that $F(t') < v_k(t_\ell \ncdots t_1)$, which concludes the proof of the claim.

	To prove Lemma~\ref{lem:FreshUndelivered}, suppose a packet $P$ is injected in round $t$. There are several cases to consider, based on the packet type.
	\begin{compactdesc}
		\item[$P$ is a type-1 packet:] In this case $P$ is injected 
		into buffer $v_1(t_\ell \ncdots t_1)$ with destination $n$. 
		Since $t' \geq t$, we have 
		\begin{align*}
		F(t')~ \leq~ F(t)%\\
			  & = v_1(t_\ell \ncdots t_1) \\
			  & = \sum_{j=1}^{\ell} \left[ (j+1)m^j - (t_j+1)jm^{j-1} \right]\\
			  & \leq \sum_{j=1}^{\ell} \left[ (j+1)m^j - jm^{j-1} \right]\\
			  & = (\ell+1)m^\ell - 1%\\
			  ~=~ n-1 %\\
			  ~<~ P(t')~.
	    \end{align*}
		\item[$P$ is a type-$(\ell+1)$ packet:] In this case $P$ is 
		injected into 
		buffer 0 with destination $v_\ell(t_\ell \ncdots t_1)$. By the 
		definition of $v_\ell(t_\ell \ncdots t_1)$, and since $t_\ell 
		\in \set{0,\ldots,m-1}$, packet $P$'s destination is 
		$(\ell+1)m^\ell - (t_\ell+1)\ell m^{\ell-1} \geq (\ell+1)m^\ell 
		- \ell m^{\ell} = m^{\ell}$. As packet $P$ moves at most one 
		hop per round, at least $m^\ell$ rounds elapse before $P$ 
		reaches its destination, i.e., $t' - t \geq m^\ell$. This means 
		that, in the base-$m$ representations of $t$ and $t'$, we have 
		$t'_\ell > t_\ell$, so by Claim~\ref{claim:FMoves}, $F(t') < 
		v_k(t_\ell \ncdots t_1) = P(t')$. 
		\item[$P$ is a type-$k$ packet with $k \in 
		\set{2,\ldots,\ell}$:] In this case $P$ is injected into buffer 
		$v_k(t_\ell\ncdots t_1)$ with destination 
		$v_{k-1}(t_\ell\ncdots t_1)$. The distance between these two 
		buffers is
		\begin{align*}
		 \sum_{j=k-1}^{\ell} \left[ (j+1)m^j - (t_j+1)jm^{j-1} \right] 
		 - \sum_{j=k}^{\ell} &\left[ (j+1)m^j - (t_j+1)jm^{j-1} 
		 \right]\\
		= &  ~km^{k-1} - (t_{k-1}+1)(k-1)m^{k-2}\\
		\geq & ~km^{k-1} - m(k-1)m^{k-2}\\
		= & ~m^{k-1}~.
		\end{align*}
	\end{compactdesc} 
	 As packet $P$ moves at most one hop per round, at least $m^{k-1}$ rounds elapse before $P$ reaches its destination, i.e., $t' - t \geq m^{k-1}$. This means that, in the base-$m$ representations of $t$ and $t'$, there exists a $k' \geq k-1$ such that $t'_{k'} > t_{k'}$ and $t'_i = t_i$ for all $i \in \set{k'+1,\ldots,\ell}$. So by Claim~\ref{claim:FMoves}, $F(t') < v_{k'}(t_\ell \ncdots t_1) \leq v_{k-1}(t_\ell \ncdots t_1) = P(t')$.
	 
	 In all cases, we showed that $P(t') > F(t')$, which means that $P$ is stale in round $t'$, as desired.
\end{proof}
\begin{proof}[Proof of Lemma~\ref{lem:boundstale}]
		As the buffer $F(t)$ can forward at most one packet per round, the number of packets that can go $\alpha$-stale in any interval of $\tau \geq 0$ rounds is at most $\tau$.
		
	  	Next, note that if $t$ is not the last round of the $t_\ell \ncdots t_1$-phase, then rounds $t$ and $t+1$ belong to the $t_\ell \ncdots t_1$-phase, which implies that $F(t+1) = v_1(t_\ell\ncdots t_1) = F(t)$. Therefore, condition 2 of Lemma~\ref{lem:staletypes} cannot be satisfied, so no packets become $\beta$-stale at the end of round $t$.
	  	
	  	Next, suppose that $t$ is the last round of the $t_\ell\ncdots t_1$-phase, and let $k$ be the smallest integer in $\set{0,\ldots,\ell-1}$ such that $t_{k+1} < m-1$. Namely, round $t$ belongs to the $t_\ell \ncdots t_{k+1}(m-1)\ncdots (m-1)$-phase, and round $t+1$ belongs to the $t_\ell \ncdots (t_{k+1}+1)0\ncdots 0$-phase. 
	  	
	  	By definition, the value of $F(t)$ is
	  	\begin{align*}
	  	v_1(t_\ell \ncdots t_{k+1}(m-1)\ncdots (m-1)) & = \sum_{j=1}^{k} \left[ (j+1)m^j - mj m^{j-1} \right]  +  \sum_{j=k+1}^{\ell} \left[ (j+1)m^j - (t_j+1)j m^{j-1} \right]\\
	  	& = \sum_{j=1}^{k} m^j   +  \sum_{j=k+1}^{\ell} \left[ (j+1)m^j 
	  	- (t_j+1)j m^{j-1} \right]~.
	  	\end{align*}
	  	Further, the value of $F(t+1)$ is
	  	\begin{align*}
	  	v_1(t_\ell \ncdots (t_{k+1}+1)0\ncdots 0) & = \sum_{j=1}^{k} \left[ (j+1)m^j - j m^{j-1} \right]  +  \sum_{j=k+1}^{\ell} \left[ (j+1)m^j - (t_j+1)j m^{j-1} \right] - (k+1)m^k\\
	  	& = ((k+1)m^k - 1)  +  \sum_{j=k+1}^{\ell} \left[ (j+1)m^j - (t_j+1)j m^{j-1} \right] - (k+1)m^k\\
	  	& = -1+\sum_{j=k+1}^{\ell} \left[ (j+1)m^j - (t_j+1)j m^{j-1} 
	  	\right]~.
	  	\end{align*}
	  	Thus, $F(t+1)+1 = F(t) - \sum_{j=1}^{k} m^j = F(t) - m(m^{k}-1)/(m-1)$. By Lemma~\ref{lem:staletypes}, all packets that become $\beta$-stale at the end of round $t$ are located in buffer interval $[F(t+1)+1,F(t)]$ in round $t+1$, which means that the number of such packets is $L^{t+1}([v_{1}(t_\ell \ncdots t_1)-m(m^{k}-1)/(m-1),v_{1}(t_\ell \ncdots t_1)])$. 
\end{proof}
\begin{proof}[Proof of Lemma~\ref{lem:scenarios}]
		We assume that scenario 1 does not occur, and show that this implies the inequality in scenario 2.
		
		For any fixed $t'_\ell \in \set{0,\ldots,m-2}$ and any fixed $k \in \set{0,\ldots,\ell-1}$, the number of distinct $t'_\ell t_{\ell-1}\ncdots t_{k+1}(m-1)\ncdots (m-1)$-phases is at most $m^{\ell-1-k}$, as $t_{\ell-1},\ldots,t_{k+1}$ can take on at most $m$ values each. So, if we assume that scenario 1 does not occur, i.e., at most $((\ell+1)\rho-1)m^{k+1}/2\ell$ packets become $\beta$-stale at the end of each such phase, then the total number packets that become $\beta$-stale over all $t'_\ell t_{\ell-1}\ncdots t_1$-phases is at most
		
%		\begin{align*}
%		& \sum_{k=0}^{\ell-1} 
%(m^{\ell-1-k})\frac{((\ell+1)\rho-1)m^{k+1}}{2\ell}\\
%		= & \sum_{k=0}^{\ell-1} \frac{((\ell+1)\rho-1)m^{\ell}}{2\ell}\\
%		= & \frac{((\ell+1)\rho-1)m^{\ell}}{2}
%		\end{align*}
$$
\sum_{k=0}^{\ell-1} (m^{\ell-1-k})\frac{((\ell+1)\rho-1)m^{k+1}}{2\ell}
~=~ \sum_{k=0}^{\ell-1} \frac{((\ell+1)\rho-1)m^{\ell}}{2\ell}
~=~ \frac{((\ell+1)\rho-1)m^{\ell}}{2}~.
$$
		
		Next, for fixed $t'_\ell$, there are $m^{\ell-1}$ $t'_\ell 
		t_{\ell-1}\ncdots t_1$-phases, and each consists of $m$ rounds, 
		so the number of rounds that elapse over all such phases is 
		$m^\ell$. By Lemma~\ref{lem:boundstale}, the total number of 
		packets that become $\alpha$-stale over all $t'_\ell 
		t_{\ell-1}\ncdots t_1$-phases is at most $m^\ell$. Therefore, 
		the total number of packets that become stale over all $t'_\ell 
		t_{\ell-1}\ncdots t_1$-phases is at most $m^\ell + 
		\frac{((\ell+1)\rho-1)m^{\ell}}{2}$. However, our injection 
		pattern injects $(\ell+1)\rho m$ fresh packets in each $t'_\ell 
		t_{\ell-1}\ncdots t_1$-phase, so the total number of fresh 
		packets injected over all $t'_\ell t_{\ell-1}\ncdots 
		t_1$-phases is $m^{\ell-1}((\ell+1)\rho m) = (\ell+1)\rho 
		m^\ell = m^\ell + ((\ell+1)\rho-1)m^{\ell}$. So, of all of the 
		packets injected from the start of the $t'_\ell 0\ncdots 
		0$-phase until the end of the $t'_\ell (m-1)\ncdots 
		(m-1)$-phase, at least $\frac{1}{2}((\ell+1)\rho-1)m^{\ell}$ of 
		them do not go stale, 
		which proves that $f(t'_\ell+1) \geq f(t'_\ell) + 
		\frac{1}{2}((\ell+1)\rho-1)m^{\ell}$. 
\end{proof}

\end{document}